\documentclass{llncs}

\usepackage{epsfig}
\usepackage{algorithmic}
\usepackage{algorithm} 
\usepackage{graphicx}
\usepackage{amssymb}
\usepackage{fullpage}
\usepackage{color}
\usepackage{amsmath}
\usepackage{mathtools}
\usepackage{cite}
\usepackage{caption}

\usepackage{versions}
\excludeversion{INUTILE}
\includeversion{LONG}\excludeversion{SHORT}\includeversion{VLONG}
\markversion{TODO}
\markversion{TODOSRINIVAS}
\excludeversion{TODOSRINIVAS}
\excludeversion{TODO}

\begin{document}

\pagestyle{headings}  
\addtocmark{                    }

\mainmatter              
\title{Synergistic Sorting, MultiSelection and Deferred Data Structures on MultiSets}

\titlerunning{Synergistic Sorting, MultiSelection and Deferred Data Structures on MultiSets}
\author{
  J\'er\'emy Barbay\inst{1} 
  \and 
  Carlos Ochoa\inst{1}
  \and
  Srinivasa Rao Satti\inst{2}
}

\institute{
  Departamento de Ciencias de la Computaci\'on, Universidad de Chile, Chile\\ 
  \email{jeremy@barbay.cl, cochoa@dcc.uchile.cl}
  \and
  Department of Computer Science and Engineering, Seoul National University, South Korea\\
  \email{ssrao@cse.snu.ac.kr}
}

\maketitle              

\begin{abstract}
  Karp et al. (1988) described Deferred Data Structures for Multisets
  as ``lazy'' data structures which partially sort data to support
  online rank and select queries, with the minimum amount of work in
  the worst case over instances of size $n$ and number of queries $q$
  fixed (i.e., the query size). Barbay et al. (2016) refined this
  approach to take advantage of the gaps between the positions hit by
  the queries (i.e., the structure in the queries). We develop new
  techniques in order to further refine this approach and to take
  advantage all at once of the structure (i.e., the multiplicities of
  the elements), the local order (i.e., the number and sizes of runs)
  and the global order (i.e., the number and positions of existing
  pivots) in the input; and of the structure and order in the sequence
  of queries. Our main result is a synergistic deferred data structure
  which performs much better on large classes of instances, while
  performing always asymptotically as good as previous solutions. As
  intermediate results, we describe two new synergistic sorting
  algorithms, which take advantage of the structure and order (local
  and global) in the input, improving upon previous results which take
  advantage only of the structure (Munro and Spira 1979) or of the
  local order (Takaoka 1997) in the input; and one new multiselection
  algorithm which takes advantage of not only the order and structure
  in the input, but also of the structure in the queries.
\begin{VLONG} We described two compressed data structures to represent a multiset taking advantage of both the local order and structure, while supporting the operators \texttt{rank} and \texttt{select} on the multiset.\end{VLONG}
\end{abstract}

\begin{center}
  \begin{minipage}{.9\textwidth}
    \noindent{\bf Keywords:} Deferred Data Structure, Divide and Conquer, Fine Grained Analysis, Quick Select, Quick Sort, Compressed Data Structure, Rank, Select.
  \end{minipage}
\end{center}

\section{Introduction}
\label{sec:intro}

Consider a \emph{multiset} $\mathcal{M}$ of size $n$ (e.g.,
$\mathcal{M}=\{1,2,3,3,4,5,6,7,8,9\}$ of size $n=10$).
%
The \emph{multiplicity} of an element $x$ of $\mathcal{M}$ is the
number $m_x$ of occurrences of $x$ in $\mathcal{M}$ (e.g.,
  $m_3=2$). We call the distribution of the multiplicities
of the elements in $\mathcal{M}$ the \emph{input
  structure}\begin{LONG}
  , and denote it by a set of pairs $(x,m_x)$ (e.g., $\{(1,1)$,
  $(2,1)$, $(3,2)$, $(4,1)$, $(5,1)$, $(6,1)$, $(7,1)$, $(8,1)$,
  $(9,1)\}$ in $\mathcal{M}$)\end{LONG}.
As early as 1976, Munro and
Spira~\cite{1976-JComp-SortingAndSearchingInMultisets-MunroSpira}
described a variant of the algorithm {\tt{MergeSort}} using counters,
which optimally takes advantage of the input structure
\begin{LONG}
  (i.e., the multiplicities of the distinct elements)
\end{LONG}
when sorting a multiset $\mathcal{M}$ of $n$ elements. Munro and Spira
measure
the ``difficulty'' of the instance in terms of the ``input structure''
by the entropy function
$\mathcal{H}(m_1, \dots, m_\sigma) =
\sum_{i=1}^\sigma{\frac{m_i}{n}}\log{\frac{n}{m_i}}$. The time
complexity of the algorithm is within
$O(n(1 + \mathcal{H}(m_1, \dots, m_\sigma))) \subseteq
O(n(1{+}\log{\sigma})) \subseteq O(n\log{n})$, where $\sigma$ is the
number of distinct elements in $\mathcal{M}$ and
$m_1, \dots, m_\sigma$ are the multiplicities of the $\sigma$ distinct
elements in $\mathcal{M}$ (such that $\sum_{i=1}^\sigma {m_i}=n$),
respectively.

Any array $\mathcal{A}$ representing a multiset lists its element in
some order, which we call the \emph{input order} and denote by a tuple
(e.g., $\mathcal{A}=(2,3,1,3,7,8,9,4,5,6)$). Maximal sorted subblocks
in $\mathcal{A}$ are a local form of input order and are called
\emph{runs}~\cite{1973-BOOK-TheArtOfComputerProgrammingVol3-Knuth}
(e.g., $\{(2,3)$, $(1,3,7,8,9)$, $(4,5,6)\}$ in $\mathcal{A}$).
As early as 1973,
Knuth~\cite{1973-BOOK-TheArtOfComputerProgrammingVol3-Knuth} described
a variant of the algorithm {\tt{MergeSort}} using a prepossessing step
taking linear time to detect \emph{runs} in the array $\mathcal{A}$\begin{LONG}
  , which he named \texttt{Natural MergeSort}.
  Mannila~\cite{1985-TCom-MeasuresOfPresortednessAndOptimalSortingAlgorithms-Mannila}
  refined the analysis of the \texttt{Natural MergeSort} algorithm to
  yield a time complexity for sorting an array $\mathcal{A}$ of size $n$
  in time within $O(n(1+\log\rho))\subseteq O(n\lg n)$, where $\rho$
  is the number of \emph{runs} in the $\mathcal{A}$\end{LONG}.
Takaoka~\cite{2009-Chapter-PartialSolutionAndEntropy-Takaoka}
described a new sorting algorithm that optimally takes advantage of the distribution of
the sizes of the runs in the array $\mathcal{A}$, which yields a time complexity within
$O(n(1+\mathcal{H}(r_1, \dots, r_{\rho}))) \subseteq
O(n(1{+}\log{\rho})) \subseteq O(n\log{n})$, where $\rho$ is the
number of runs in $\mathcal{A}$ and $r_1, \dots, r_{\rho}$ are the sizes
of the $\rho$ \emph{runs} in $\mathcal{A}$ (such that
$\sum_{i=1}^\rho {r_i}=n$), respectively.
\begin{LONG}
  It is worth noting that in 1997,
  Takaoka~\cite{1997-TR-MinimalMergesort-Takaoka} first described this
  algorithm in a technical report.
\end{LONG}

Given an element $x$ of a multiset $\mathcal{M}$ and an integer
$j\in[1..n]$\begin{LONG}
  (e.g., $\mathcal{M}=\{1,2,3,3,4,5,6,7,8,9\}$, $x=3$ and $j=4$)
\end{LONG}, the \emph{rank} $\mathtt{rank}(x)$ of $x$ is the number of elements
smaller than $x$ in $\mathcal{M}$,
\begin{LONG}
  (e.g., $\mathtt{rank}(3)=2$)
\end{LONG}
and \emph{selecting} the $j$-th element in $\mathcal{M}$ corresponds
to computing the value $\mathtt{select}(j)$ of the $j$-th smallest
element (counted with multiplicity) in $\mathcal{M}$\begin{LONG}
  (e.g., $\mathtt{select}(4)=3$)\end{LONG}.
\begin{INUTILE}
  The support of \texttt{rank} and \texttt{select} queries in
  multisets is related to the \textsc{Sorting} problem, as sorting the
  array $\mathcal{A}$ representing $\mathcal{M}$ permits supporting
  each \texttt{select} query in constant time and each \texttt{rank}
  query in $\lceil\log_2 n\rceil$ comparisons. What makes
  \textsc{Sorting} and supporting \texttt{rank} and \texttt{select}
  queries distinct is that when there are only a few queries (or when
  the multiset is dynamically updated), sorting the whole multiset is
  overkill and there are much better solutions.
\end{INUTILE}
%
As early as 1961, Hoare~\cite{1961-CACM-Quickselect-Hoare} showed how
to support \texttt{rank} and \texttt{select} queries in average linear
time, a result later improved to worst case linear time by Blum et
al.~\cite{1973-JCSS-TimeBoundsForSelection-BlubFloydPrattRivestTarjan}.
Twenty years later, Dobkin and
Munro~\cite{1981-JACM-OptimalTimeMinimalSpaceSelectionAlgorithms-DobkinMunro}
described a \textsc{MultiSelection} algorithm that supports several
\texttt{select} queries
\begin{LONG}
  in parallel
\end{LONG}
and whose running time is optimal in the worst
case over all multisets of size $n$ and all sets of $q$ queries hitting
positions in the multisets separated by \emph{gaps} (differences
between consecutive \texttt{select} queries in sorted order) of sizes
$g_0, \dots, g_q$.
\begin{LONG}
  Kaligosi et
  al.~\cite{2005-ICALP-TowardsOptimalMultopleSelection-KaligosiMehlhornMunroSanders}
  later described a variant of this algorithm which number of
  comparisons performed is within a negligible additional term of the
  optimal.
\end{LONG}
%
\begin{LONG}
  In the online context where the queries arrive one at a time,
\end{LONG}
Karp et
al.~\cite{1988-SIAM-DeferredDataStructuring-KarpMotwaniRaghavan}
further extended Dobkin and Munro's result~\cite{1981-JACM-OptimalTimeMinimalSpaceSelectionAlgorithms-DobkinMunro}\begin{SHORT}
  to the online context, where the multiple \texttt{rank} and
  \texttt{select} queries arrive one by one\end{SHORT}. Karp et
al. called their solution a \textsc{Deferred Data Structure} and
describe it as ``lazy'', as it partially sorts data, performing the
minimum amount of work necessary in the worst case over all instances
for a fixed $n$ and $q$.
Barbay et
al.~\cite{2016-JDA-NearOptimalOnlineMultiselectionInInternalAndExternalMemory-BarbayGuptaRaoSorenson}
refined this result by taking advantage of the gaps between the
positions hit by the queries (i.e., the \emph{query structure}).
\begin{INUTILE}
  Karp et al.'s
  approach~\cite{1988-SIAM-DeferredDataStructuring-KarpMotwaniRaghavan}
  is a variant of
  \texttt{MergeSort}~\cite{1973-BOOK-TheArtOfComputerProgrammingVol3-Knuth}
  while Barbay et al.'s
  approach~\cite{2016-JDA-NearOptimalOnlineMultiselectionInInternalAndExternalMemory-BarbayGuptaRaoSorenson}
  is based on \texttt{QuickSort}~\cite{1961-CACM-Quicksort-Hoare}.
\end{INUTILE}
\begin{INUTILE}
  The support for \texttt{rank} and \texttt{select} operators is
  tightly related to the task of \textsc{Sorting} or \textsc{Partially
    Sorting} multisets. Yet there are many approaches to sorting which
  take advantage of easy instances, for various definitions of
  ``easy''. We focus on two particular ones based on
  \texttt{MergeSort} and orthogonal notions of ``easiness'':
  \emph{input order} and \emph{input structure}:
\end{INUTILE}
%
This suggests the following questions:
\textbf{\begin{enumerate}
  \item Is there a sorting algorithm for multisets which takes the
    best advantage of both its \emph{input order} and its \emph{input
      structure} in a synergistic way, so that it performs as good as
    previously known solutions on all instances, and much better on
    instances where it can take advantage of both at the same time?
  \item Is there a multiselection algorithm and a deferred data structure
    for answering \texttt{rank} and \texttt{select} queries which
    takes the best advantage not only of both of those notions of
    easiness in the input, but additionally also of notions of
    easiness in the queries, such as the \emph{query structure} and
    the \emph{query order}?
  \end{enumerate}
}

We answer both questions affirmatively: 
In the context of \textsc{Sorting}, this improves upon both algorithms from Munro and Spira~\cite{1976-JComp-SortingAndSearchingInMultisets-MunroSpira} and Takaoka~\cite{2009-Chapter-PartialSolutionAndEntropy-Takaoka}. 
In the context of \textsc{MultiSelection} and \textsc{Deferred Data
  Structure} for \texttt{rank} and \texttt{select} on
\textsc{Multisets}, this improves upon Barbay et al.'s
results~\cite{2016-JDA-NearOptimalOnlineMultiselectionInInternalAndExternalMemory-BarbayGuptaRaoSorenson}
by adding 3 new measures of difficulty (input order, input structure
and query order) to the single one previously considered (query
structure). Even though the techniques used by our algorithms are
known, the techniques used to refine the analysis of these algorithms
to show that they improve the state of the art are complex.
Additionally, 
we correct the analysis of the \texttt{Sorted Set Union} algorithm by Demaine et al.~\cite{2000-SODA-AdaptiveSetIntersectionsUnionsAndDifferences-DemaineLopezOrtizMunro} (Section~\ref{sec:dlm-sorting}), and
we define a simple yet new notion of ``global'' input order (Section~\ref{sec:global})\begin{LONG}, formed by the number of preexisting pivot positions in the input (e.g. $(3,2,1,6,5,4)$ has one pre-existing pivot position in the middle)\end{LONG}, not  mentioned in previous surveys~\cite{1992-ACMCS-ASurveyOfAdaptiveSortingAlgorithms-EstivillCastroWood,1992-ACJ-AnOverviewOfAdaptiveSorting-MoffatPetersson} nor extensions \cite{2013-TCS-OnCompressingPermutationsAndAdaptiveSorting-BarbayNavarro}.

We present our results incrementally, each building on the previous
one, such that the most complete and complex result is in Section~\ref{sec:online-synergy-defer}. In Section~\ref{sec:synergy-sorting} we describe how to measure the interaction of the order (local and global) with the structure in the input, and two new synergistic \textsc{Sorting} algorithms based on distinct paradigms (i.e., merging vs splitting) which take advantage of both the input order and structure\begin{LONG} in order to sort the multiset in less time than traditional solutions based on one of those, at most\end{LONG}. We refine the second of those results in Section~\ref{sec:synergy-deferr-data} with the analysis of a \textsc{MultiSelection} algorithm which takes advantage of not only the order and structure in the input, but also of the \emph{query structure}, in the offline setting. In Section~\ref{sec:online-synergy-defer} we analyze an online \textsc{Deferred Data Structure} taking advantage of the order and structure in the input on one hand, and of the order and structure in the queries on the other hand\begin{LONG}, in a synergistic way\end{LONG}.\begin{VLONG} As an additional result, we describe in Section~\ref{sec:compressed} two compressed data structures to represent a multiset taking advantage of both the input order and structure, while supporting the operators \texttt{rank} and \texttt{select} on the multiset.\end{VLONG} We conclude with a discussion of our results in Section~\ref{sec:discussion}.

\section{Sorting Algorithms}\label{sec:synergy-sorting}

We review in Section~\ref{sec:review} the algorithms \texttt{MergeSort
  with Counters} described by Munro and
Spira~\cite{1976-JComp-SortingAndSearchingInMultisets-MunroSpira} and
\texttt{Minimal MergeSort} described by
Takaoka~\cite{2009-Chapter-PartialSolutionAndEntropy-Takaoka}, each
takes advantage of distinct features in the input.\begin{LONG} In
  Section~\ref{sec:comp-betw-algor} we show that the algorithm
  \texttt{MergeSort with Counters} is incomparable with the algorithm
  \texttt{Minimal MergeSort}, in the sense that none performs always
  better than the other (by describing families of instances where
  \texttt{MergeSort with Counters} performs better than
  \texttt{Minimal MergeSort} on one hand, and families of instances
  where the second one performs better than the first one on the other
  hand), and some simple modifications and combinations of these
  algorithms, which still do not take full advantage of both the order
  (local and global) and structure in the input.\end{LONG} In Sections
\ref{sec:dlm-sorting} and \ref{sec:ttqu-synergy-sort}, we describe two
synergistic \textsc{Sorting} algorithms, which never perform worse
than \begin{SHORT}\texttt{MergeSort with Counters} and \texttt{Minimal MergeSort}\end{SHORT}\begin{LONG} the algorithms presented in Section~\ref{sec:review} and Section~\ref{sec:comp-betw-algor}\end{LONG}, and perform much better on some large classes of instances by taking advantage of both the order (local and global) and the structure in the input, in a synergistic way.

\begin{LONG}
  \subsection{Previously Known Input-structure and Local Input-order
    Algorithms}\label{sec:review}
\end{LONG}

\begin{SHORT}
  \subsection{Known  Algorithms}\label{sec:review}
\end{SHORT}

The algorithm \texttt{MergeSort with Counters} described by Munro and
Spira~\cite{1976-JComp-SortingAndSearchingInMultisets-MunroSpira} is
an adaptation of the traditional sorting algorithm \texttt{MergeSort}
that optimally takes advantage of the structure in the input when
sorting a multiset $\mathcal{M}$ of size $n$. The algorithm divides
$\mathcal{M}$ into two equal size parts, sorts both parts recursively,
and then merges the two sorted lists. When two elements of same value $v$ are
found, one is thrown away and a counter holding the number of occurrences
of $v$ is updated. Munro and Spira measure the
``difficulty'' of the instance in terms of the ``input structure'' by
the entropy function
$\mathcal{H}(m_1, \dots, m_\sigma) =
\sum_{i=1}^\sigma{\frac{m_i}{n}}\log{\frac{n}{m_i}}$.  The time
complexity of the algorithm is within
$O(n(1 + \mathcal{H}(m_1, \dots, m_\sigma))) \subseteq
O(n(1{+}\log{\sigma})) \subseteq O(n\log{n})$, where $\sigma$ is the
number of distinct elements in $\mathcal{M}$ and
$m_1, \dots, m_\sigma$ are the multiplicities of the $\sigma$ distinct
elements in $\mathcal{M}$ (such that $\sum_{i=1}^\sigma {m_i}=n$),
respectively.

The algorithm \texttt{Minimal MergeSort} described by
Takaoka~\cite{2009-Chapter-PartialSolutionAndEntropy-Takaoka}
optimally takes advantage of the local order in the input, as measured
by the decomposition into runs when sorting an array $\mathcal{A}$ of
size $n$.  The main idea is to detect the runs first and then merge
them pairwise\begin{LONG}, using a \texttt{MergeSort}-like
  step\end{LONG}. The runs are detected in linear time\begin{LONG} by
  a scanning process identifying the positions $i$ in $\mathcal{A}$
  such that $\mathcal{A}[i] > \mathcal{A}[i+1]$\end{LONG}. Merging the
two shortest runs at each step further reduces the number of
comparisons, making the running time of the merging process adaptive
to the entropy of the sequence of the sizes of the runs.
\begin{LONG}
  The merging process is then represented by a tree with the shape of
  a
  Huffman~\cite{1952-IRE-AMethodForTheInstructionOfMinimumRedundancyCodes-Huffman}
  tree, built from the distribution of the sizes of the \emph{runs}.
\end{LONG}
If the array $\mathcal{A}$ is formed by $\rho$ runs and
$r_1, \dots, r_{\rho}$ are the sizes of the $\rho$ runs (such that
$\sum_{i=1}^\rho {r_i}=n$), then the algorithm sorts $\mathcal{A}$ in time within
$O(n(1+\mathcal{H}(r_1, \dots, r_{\rho}))) \subseteq
O(n(1{+}\log{\rho})) \subseteq O(n\log{n})$.

\begin{SHORT}
The algorithm \texttt{MergeSort with Counters} is incomparable with
the algorithm \texttt{Minimal MergeSort}, in the sense that neither
one performs always better than the other. Simple modifications and
combinations of these algorithms do not take full advantage of both
the order (local and global) and structure in the input (see the
extended
version~\cite{2016-ARXIV-SynergisticSortingAndDeferredDataStructuresOnMultiSets-BarbayOchoaSatty}
for detailed counter examples).
\end{SHORT}

\begin{LONG}
  \subsection{Comparison Between Sorting
    Algorithms}\label{sec:comp-betw-algor}

  \paragraph*{Example 1: $\left\langle 1,2,1,2, \dots,1,2 \right\rangle$}~\\

  On this family of instances, the algorithm \texttt{Minimal
    MergeSort} detects $\frac{n}{2}$ runs that merges two at a
  time. The time complexity of \texttt{Minimal MergeSort} is within
  $\Theta(n\log{n})$. On the other hand, the algorithm
  \texttt{MergeSort with Counters} recognizes the multiplicity of the
  elements with values $1$ and $2$. In every merging step, the
  resulting set is always $\{1,2\}$. Therefore, the number of elements
  is reduced by half at each step, which yields a complexity linear in
  the size of the instance for sorting such instances.

  \paragraph*{Example 2: $\langle 1,2, \dots, n \rangle$}~\\

  On this family of instances, the algorithm \texttt{Minimal
    MergeSort} detects in time linear in the size $n$ of the input
  that the sequence is already sorted, and in turn it finishes. On the
  other hand, the running time of the algorithm \texttt{MergeSort with
    Counters} is within $\Theta(n\log{n})$ because this algorithm
  would perform at least as many operations as {\tt{MergeSort}}.

  The algorithm \texttt{Parallel Minimal Counter MergeSort} runs both
  algorithms in parallel, when one of these algorithms manages to sort
  the sequence, the algorithm {\tt{Parallel Minimal Counter MergeSort}} returns the
  sorted sequence and then finishes. The time complexity of this
  algorithm in any instance is twice the minimum of the complexities
  of \texttt{Minimal MergeSort} and \texttt{MergeSort with Counters}
  for this instance. The algorithm \texttt{Parallel Minimal Counter
    MergeSort} needs to duplicate the input in order to run both
  algorithms in parallel.

  Combining the ideas of identifying and merging runs from
  Takaoka~\cite{2009-Chapter-PartialSolutionAndEntropy-Takaoka} with
  the use of counters by Munro and
  Spira~\cite{1976-JComp-SortingAndSearchingInMultisets-MunroSpira},
  we describe the {\tt{Small vs Small Sort}} algorithm to sort a
  multiset. It identifies the runs using the same linear scanning as
  \texttt{Minimal MergeSort} and associates counters to the elements
  in the same way that \texttt{MergeSort with Counters} does. Once the
  runs are identified, this algorithm initializes a heap with them
  ordered by sizes. At each merging step the two shorter runs are
  selected for merging and both are removed from the heap. The pair is
  merged and the resulting run is inserted into the heap. The process
  is repeated until only one run is left and the sorted sequence is
  known. The {\tt{Small vs Small Sort}} algorithm is adaptive to the
  sizes of the resulting runs in the merging process.

  The time complexity of {\tt{Small vs Small Sort}} for all instances
  is within a constant factor of the time complexity of {\tt{Parallel
      Minimal Counter MergeSort}}; but the next example shows that
  there are families of instances where {\tt{Small vs Small Sort}}
  performs better than a constant factor of the time complexity of
  {\tt{Parallel Minimal Counter MergeSort}}.

  \paragraph*{Example 3: $\langle 1,2, \dots, \sigma,1,2, \dots, \sigma \dots, 1,2,\dots,\sigma \rangle$}~\\

  In this family of instances, there are $\rho$ runs, each of size
  $\sigma$. The complexity of \texttt{MergeSort with Counters} for
  this instance is within $\Theta(\rho\sigma\log{\sigma})$, while the
  complexity of \texttt{Minimal MergeSort} for this instance is within
  $\Theta(\rho\sigma\log{\rho})$. On the other hand, the complexity
  {\tt{Small vs Small Sort}} for this instance is better, within
  $\Theta(\rho\sigma)$: at each level of the binary tree representing
  the merging order, the sum of the sizes of the runs is halved.

  Even though {\tt{Small vs Small Sort}} is adaptive to the sizes of
  the resulting runs, it does not take advantage of the fact that
  there may exist a pair of runs that can be merged very quickly, but
  it needs linear time in the sum of the sizes to merge one of them
  when paired with another run of the same size. We show this
  disadvantage in the Example $4$ at the end of the next section.
\end{LONG}

In the following sections we describe two sorting algorithms that take
the best advantage of both the order (local and global) and structure
in the input all at once when sorting a multiset. The first one is a
straightforward application of previous results, while the second one
prepares the ground for the \textsc{MultiSelection} algorithm (Section~\ref{sec:synergy-deferr-data}) and the
\textsc{Deferred Data Structures} (Section~\ref{sec:online-synergy-defer}), which take advantage of the order
(local and global) and structure in the data and of the order and
structure in the queries.

\subsection{``Kind-of-new'' Sorting Algorithm \texttt{DLM
    Sort}}\label{sec:dlm-sorting}

In 2000, Demaine et
al.~\cite{2000-SODA-AdaptiveSetIntersectionsUnionsAndDifferences-DemaineLopezOrtizMunro}
described the algorithm \texttt{DLM Union}, an instance optimal algorithm that
computes the union of $\rho$ sorted sets.
\begin{LONG}
  It inserts the smallest element of each set in a heap. At each step,
  it deletes from the heap all the elements whose values are equal to
  the minimum value of the heap. If more than one element is deleted,
  it knows the multiplicity of this value in the union of the sets. It
  then adds to the heap the elements following the elements of minimum
  value of each set that contained the minimum value. If there is only
  one minimum element, it extracts from the heap the second minimum
  and executes a doubling
  search~\cite{1976-IPL-AnAlmostOptimalAlgorithmForUnboundedSearching-BentleyYao}
  in the set where the minimum belongs for the value of the second
  minimum. Once it finds the insertion rank $r$ of the second minimum
  (i.e., number of elements smaller than the second minimum in the set
  that contains the minimum), it also knows that the multiplicity of
  all elements whose positions are before $r$ in the set that contain
  the minimum are 1 in the union of the $\rho$ sets.
  \begin{INUTILE}
    These elements are discarded from future iterations of the
    algorithm.
  \end{INUTILE}
The process is repeated until all elements
  are discarded.
\end{LONG}\begin{SHORT}The algorithm scans the sets from left to right identifying blocks of
  consecutive elements in the sets that are also consecutive in the
  sorted union (see Figure~\ref{fig:instance} for a graphical
  representation of such a decomposition on a particular
  instance of the \textsc{Set Union} problem). In a minor way we refine their analysis as follow:
  
\end{SHORT}
\begin{LONG}
  
  The time complexity of the \texttt{DLM Union} algorithm is measured
  in terms of the number and sizes of blocks of consecutive elements
  in the sets that are also consecutive in the sorted union (see
  Figure~\ref{fig:instance} for a graphical representation of such a
  decomposition on a particular instance of the \textsc{Sorted Set Union}
  problem). The sizes of these blocks are referred to as \emph{gaps}
  in the analysis of the algorithm.\end{LONG} These blocks induce a
partition $\pi$ of the output into intervals such that any singleton
corresponds to a value that has multiplicity greater than $1$ in the
input, and each other interval corresponds to a block as defined
above. Each member $i$ of $\pi$ has a value $m_i$ associated with it:
if the member $i$ of $\pi$ is a block, then $m_i$ is $1$, otherwise,
if the member $i$ of $\pi$ is a singleton corresponding to a value of
multiplicity $q$, then $m_i$ is $q$.
Let $\chi$ be the size of $\pi$.
If the instance is formed by $\delta$ blocks of sizes $g_1, \dots, g_{\delta}$ such that these blocks induce a partition $\pi$ whose members have values $m_1, \dots, m_{\chi}$, we express the time complexity of \texttt{DLM Union} as within $\Theta(\sum^{\delta}_{i=1}\log g_i + \sum^{\chi}_{i=1}\log{\binom{\rho}{m_i}})$.
This time complexity is within a constant factor of the complexity of any other algorithm computing the union of these sorted sets (i.e., the algorithm is instance optimal).

\begin{minipage}[c]{.45\textwidth}
  \centering
  \includegraphics[scale=1.2]{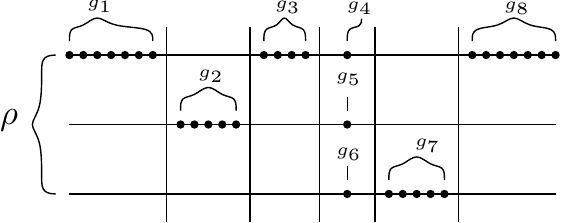}
\end{minipage}\hfill
\begin{minipage}[c]{.45\textwidth}
  \captionof{figure}{An instance of the \textsc{Sorted Set Union} problem
    with $\rho=3$ sorted sets. In each set, the entry $\mathcal{A}[i]$ is
    represented by a point of $x$-coordinate $\mathcal{A}[i]$. The blocks
    that form the sets are noted. The blocks $g_4,g_5$ and $g_6$ are of
    size 1 because they correspond to elements of equal value and they
    induce the $4$-th member of the partition $\pi$ with value $m_4$
    equals $3$. The vertical bars separate the members of $\pi$.}
  \label{fig:instance}
\end{minipage}

We adapt the \texttt{DLM Union} algorithm for sorting a multiset.  The algorithm \texttt{DLM Sort} detects the runs first through a linear scan and then applies the algorithm \texttt{DLM Union}. After that, transforming the output of the union algorithm to yield the sorted multiset takes only linear time. The following corollary follows from our refined analysis above:

\begin{corollary}
  Given a multiset $\mathcal{M}$ of size $n$ formed by $\rho$ runs and
  $\delta$ blocks of sizes $g_1, \dots, g_{\delta}$ such that these
  blocks induce a partition $\pi$ of size $\chi$ of the output whose
  members have values $m_1, \dots, m_{\chi}$, the algorithm {\tt{DLM Sort}} performs
  within
  $O(n + \sum^{\delta}_{i=1}\log g_i +
  \sum^{\chi}_{i=1}\log{\binom{\rho}{m_i}})$ data comparisons. This
  number of comparisons is optimal in the worst case over multisets of
  size $n$ formed by $\rho$ runs and $\delta$ blocks of sizes
  $g_1, \dots, g_{\delta}$ such that these blocks induce a partition
  $\pi$ of size $\chi$ of the output whose members have values
  $m_1, \dots, m_{\chi}$.
\end{corollary}

\begin{LONG}
There are families of instances where the algorithm \texttt{DLM Sort} behaves significantly better than the algorithm \texttt{Small vs Small Sort}. Consider for instance the following example:
  \paragraph*{Example 4:
    $\left\langle \frac{\rho-1}{\rho}n+1, \dots, n,
      \frac{\rho-2}{\rho}n+1, \dots, \frac{\rho-1}{\rho}n, \dots, 1,
      \dots, \frac{1}{\rho}n \right\rangle$}~\\

  In this family of instances, there are $\rho$ runs of size
  $\frac{n}{\rho}$ each. The runs are pairwise disjoint and the
  elements of each run are consecutive in the sorted set. The time
  complexity of the algorithm \texttt{Small vs Small Sort} in this instances is
  within $\Theta(n\log{\rho})$, while the time complexity of the algorithm 
  \texttt{DLM Sort} is within
  $\Theta\left(n + \rho\log\rho + \rho\log{\frac{n}{\rho}}\right) =
  \Theta(n + \rho\log{n})$ (which is better than $\Theta(n\log{\rho})$
  for $\rho \in o(n)$).
\end{LONG}

While the algorithm \texttt{DLM Sort} answers the Question 1 from Section~\ref{sec:intro}, it does not yield a \textsc{MultiSelection} algorithm nor a \textsc{Deferred Data Structure} answering Question 2. In the following section we describe another sorting algorithm that also optimally takes advantage of the local order and structure in the input, but which is based on a distinct paradigm, more suitable to such extensions.

\subsection{New Sorting Algorithm {\texttt{Quick Synergy
      Sort}}}\label{sec:ttqu-synergy-sort}

Given a multiset $\mathcal{M}$, the algorithm \texttt{Quick Synergy
  Sort} identifies the \emph{runs} in linear time through a scanning
process. \begin{LONG} As indicated by its name, the algorithm is
  directly inspired by the
  \texttt{QuickSort}~\cite{1961-CACM-Quicksort-Hoare}
  algorithm.\end{LONG} It computes a pivot $\mu$, which is the median
of the set formed by the middle elements of each run, and partitions
each \emph{run} by the value of $\mu$. This partitioning process takes advantage of
the fact that the elements in each \emph{run} are already sorted. It
then recurses on the elements smaller than $\mu$ and on the elements
greater than $\mu$. (See Algorithm~\ref{alg:qss} for a more formal
description).

\begin{definition}[Median of the middles]
  Given a multiset $\mathcal{M}$ formed by runs, the ``\emph{median of
    the middles}'' is the median element of the set formed by the
  middle elements of each run.
\end{definition}

\begin{algorithm} 
  \caption{\texttt{Quick Synergy Sort}} 
  \label{alg:qss} 
  \begin{algorithmic}[1] 

    \REQUIRE{A multiset $\mathcal{M}$ of size $n$} \ENSURE{A sorted sequence of
      $\mathcal{M}$}

    \STATE Compute the $\rho$ runs of respective sizes
    $(r_i)_{i\in[1..\rho]}$ in $\mathcal{M}$ such that
    $\sum^{\rho}_{i=1} r_i = n$;
    \STATE Compute the median $\mu$ of
    the middles of the runs, note $j\in[1..\rho]$ the run containing
    $\mu$;
    \STATE Perform doubling searches for the value $\mu$ in all
    runs except the $j$-th, starting at both ends of the runs in
    parallel;
    \STATE Find the maximum $\max_\ell$ (minimum $\min_r$)
    among the elements smaller (resp., greater) than $\mu$ in all runs
    except the $j$-th;
    \STATE Perform doubling searches for the values
    $\max_\ell$ and $\min_r$ in the $j$-th run, starting at the
    position of $\mu$;
    \STATE Recurse on the elements smaller than o
    equal to $\max_\ell$ and on the elements greater than or equal to
    $\min_r$.
  \end{algorithmic}
\end{algorithm}

The number of data comparisons performed by the algorithm
\texttt{Quick Synergy Sort} is asymptotically the same as the number
of data comparisons performed by the algorithm \texttt{DLM Sort}
described in the previous section. We divide the proof into two
lemmas. We first bound the
\begin{LONG}
  overall
\end{LONG}
number of data comparisons performed by all the
doubling searches of the algorithm \texttt{Quick Synergy Sort}.

\begin{lemma}\label{lem:blocks}
  Let $g_{1}, \dots, g_{k}$ be the sizes of the $k$ blocks that form
  the $r$-th run. The overall number of data comparisons performed by
  the doubling searches of the algorithm  \texttt{Quick Synergy Sort} to find the
  values of the medians of the middles in the $r$-th run is within
  $O(\sum_{i=1}^{k}\log{g_{i}})$.
\end{lemma}

\begin{proof}
  Every time the algorithm finds the insertion rank of one of the medians of
  the middles in the $r$-th run, it partitions the run by a position
  separating two blocks. The doubling search steps can be represented
  as a tree\begin{LONG} (see Figure~\ref{fig:tree} for a tree
    representation of a particular instance)\end{LONG}. Each
  node of the tree corresponds to a step. Each
  internal node has two children, which correspond to the two
  subproblems into which the step partitions the
  run. The cost of the step is less than four times the logarithm of
  the size of the child subproblem with smaller size, because of the
  two doubling searches in parallel. The leaves of the tree correspond
  to the blocks themselves.

  We prove that at each step, the total cost is bounded by eight times
  the sum of the logarithms of the sizes of the leaf subproblems. This
  is done by induction over the number of steps. If the number of
  steps is zero then there is no cost. For the inductive step, if the
  number of steps increases by one, then a new doubling search step is
  done and a leaf subproblem is partitioned into two new
  subproblems. At this step, a leaf of the tree is transformed into an
  internal node and two new leaves are created. Let $a$ and $b$ such
  that $a\leq b$ be the sizes of the new leaves created. The cost of
  this step is less than $4\log{a}$. The cost of all the steps then
  increases by $4\lg a$, and hence the sum of the logarithms of the
  sizes of the leaves increases by $8(\lg a + \lg b) -
  8\lg({a+b})$. But if $a \ge 4$ and $b \ge a$, then
  $2\lg(a+b) \le \lg a + 2\lg b$. The result
  follows.\begin{LONG}\qed\end{LONG}
\end{proof}

\begin{LONG}
  \begin{minipage}[c]{.45\textwidth}
    \centering
    \includegraphics[scale=1.1]{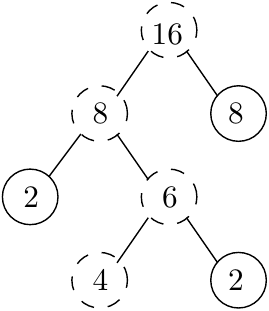}
  \end{minipage}\hfill
  \begin{minipage}[c]{.45\textwidth}
    \captionof{figure}{The tree that represents the doubling search
      steps for a run composed of four blocks of respective sizes
      $2,4,2,8$. The size of the subproblem is noted in each node. At
      each subproblem, the cost of the step is the logarithm of the
      size of the subproblem of the solid child.}
    \label{fig:tree}
  \end{minipage}
\end{LONG}
%

\begin{INUTILE}
  When the algorithm finds the insertion rank of the median $\mu$ of
  the middles in each run, it also finds the divisions between the
  blocks in the runs, or yields complete runs to the left or to the
  right of $\mu$.
\end{INUTILE}

\begin{LONG}
  The step that computes the median $\mu$ of the middles of $\rho$
  runs and the step that finds the maximum $\max_\ell$ (minimum
  $\min_r$) among the elements smaller (resp., greater) than $\mu$ of
  $\rho$ runs performs linear in $\rho$ data comparisons.
\end{LONG}
As shown in the following lemma, the overall number of data
comparisons performed during
\begin{LONG}
  these steps
\end{LONG}
\begin{SHORT}the computation of the medians of the middles
  \begin{INUTILE}
    and during the finding of the maximum (minimum) among the elements
    smaller (resp., greater) than the medians
  \end{INUTILE}
\end{SHORT}is within $O(\sum^{\chi}_{i=1}\log{\binom{\rho}{m_i}})$, where
$m_1, \dots, m_{\chi}$ are the values of the member of the partition
$\pi$ (see Section~\ref{sec:dlm-sorting} for the definition of $\pi$)
and $\rho$ is the number of runs in $\mathcal{M}$.

Consider the instance depicted in Figure~\ref{fig:clue} for an example
illustrating where the term $\log{\binom{\rho}{m_v}}$ comes from. In
this instance, there is a value $v$ that has multiplicity $m_v>1$ in
$\mathcal{M}$ and the rest of the values have multiplicity $1$. The
elements with value $v$ are present at the end of the last $m_v$ runs
and the rest of the runs are formed by only one block. The elements of
the $i$-th run are greater than the elements of the $(i+1)$-th
run. During the computation of the medians of the middles, the number
of data comparisons that involve elements of value $v$ is within
$O(\log{\binom{\rho}{m_v}})$. The algorithm computes the median $\mu$
of the middles and partitions the runs by the value of $\mu$. In the recursive call
that involves elements of value $v$, the number of runs is reduced by
half. This is repeated until one occurrence of $\mu$ belongs to one of
the last $m_v$ runs. The number of data comparisons that involve
elements of value $v$ up to this step is within
$O(m_v\log{\frac{\rho}{m_v}})=O(\log{\binom{\rho}{m_v}})$, where
$\log{\frac{\rho}{m_v}}$ corresponds to the number of steps where
$\mu$ does not belong to the last $m_v$ runs. The next recursive call
will necessarily choose one element of value $v$ as the median of the
middles.

\begin{minipage}[c]{.45\textwidth}
  \centering
  \includegraphics[scale=1]{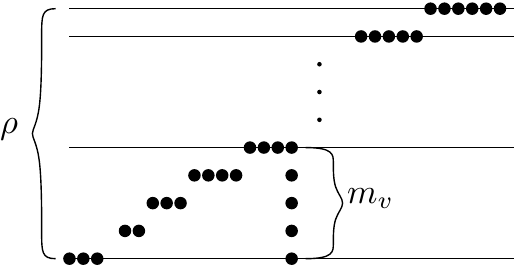}
\end{minipage}\hfill
\begin{minipage}[c]{.45\textwidth}
  \captionof{figure}{A multiset $\mathcal{M}$ formed by $\rho$
    runs. Each entry $\mathcal{M}[i]$ is represented by a point of
    $x$-coordinate $\mathcal{M}[i]$. There is an element of multiplicity $m_v$
    present in the last $m_v$ runs and the rest of the runs are formed
    by only one block.}
  \label{fig:clue}
\end{minipage}

\begin{lemma}\label{lem:delta}
  Let $\mathcal{M}$ be a multiset formed by $\rho$ runs and $\delta$
  blocks such that these blocks induce a partition $\pi$ of the output
  of size $\chi$ whose members have values $m_1, \dots,
  m_{\chi}$. Consider the steps that compute the medians of the
  middles and the steps that find the elements $\max_\ell$ and
  $\min_r$ in the algorithm  \texttt{Quick Synergy Sort}, the overall number of data
  comparisons performed during these steps is within
  $O(\sum^{\chi}_{i=1}\log{\binom{\rho}{m_i}})$.
\end{lemma}

\begin{proof}
  We prove this lemma by induction over the size of $\pi$ and
  $\rho$. The number of data comparisons performed by one of these
  steps is linear in the number of runs in the sub-instance (i.e.,
  ignoring all the empty sets of this sub-instance). Let
  $\mathcal{T}(\pi,\rho)$ be the overall number of data comparisons
  performed during the steps that compute the medians of the middles
  and during the steps that find the elements $\max_\ell$ and $\min_r$
  in the algorithm  \texttt{Quick Synergy Sort}. We prove that
  $\mathcal{T}(\pi, \rho) \le
  \sum^{\chi}_{i=1}m_i\log{\frac{\rho}{m_i}} - \rho$. Let $\mu$ be the
  first median of the middles computed by the algorithm.  Let $\ell$
  and $r$ be the number of runs that are completely to the left and to
  the right of $\mu$, respectively. Let $b$ be the number of runs that
  are split in the doubling searches for the value $\mu$ in
  all runs. Let $\pi_\ell$ and $\pi_r$ be the partitions induced by
  the blocks yielded to the left and to the right of $\mu$,
  respectively. Then,
  $\mathcal{T}(\pi, \rho) = \mathcal{T}(\pi_\ell, \ell+b) +
  \mathcal{T}(\pi_r, r+b) + \rho$ because of the two recursive calls
  and the step that computes $\mu$. By Induction Hypothesis,
  $\mathcal{T}(\pi_\ell,\ell+b) \le
  \sum^{\chi_\ell}_{i=1}m_i\log{\frac{\ell + b}{m_i}} - \ell - b$ and
  $\mathcal{T}(\delta_r, r+b) \le
  \sum^{\chi_r}_{i=1}m_i\log{\frac{r+b}{m_i}} - r - b$. Hence, we
  need to prove that
  $\ell+r \le \sum^{\chi_\ell}_{i=1}m_i \log\left({1 +
      \frac{r}{\ell+b}}\right) + \sum^{\chi_r}_{i=1}m_i\left({1 +
      \frac{\ell}{r+b}}\right)$, but this is a consequence of
  $\sum^{\chi_\ell}_{i=1}m_i \ge \ell+b, \sum^{\chi_r}_{i=1}m_i \ge
  r+b$ (the number of blocks is greater than or equal to the number of runs);
  $\ell \le r+b, r\le\ell + b$ (at least $\frac{\rho}{2}$ runs are
  left to the left and to the right of $\mu$); and
  $\log\left({1 + \frac{y}{x}}\right)^x \ge y$ for
  $y \le x$.\begin{LONG}\qed\end{LONG}
\end{proof}

Consider the step that performs doubling searches for the values
$\max_\ell$ and $\min_r$ in the run that contains the median $\mu$ of
the middles, this step results in the finding of the block $g$ that
contains $\mu$ in at most $4\log{|g|}$ data comparisons, where $|g|$
is the size of $g$.
Combining Lemma~\ref{lem:blocks} and Lemma~\ref{lem:delta} yields an
upper bound on the number of data comparisons performed by
the algorithm \texttt{Quick Synergy Sort}:

\begin{theorem}
  Let $\mathcal{M}$ be a multiset of size $n$ formed by $\rho$ runs
  and $\delta$ blocks of sizes $g_1, \dots, g_{\delta}$ such that
  these blocks induce a partition $\pi$ of the output of size $\chi$
  whose members have values $m_1, \dots, m_{\chi}$. The algorithm
  \texttt{Quick Synergy Sort} performs within
  $O(n + \sum^{\delta}_{i=1} \log g_i +
  \sum^{\chi}_{i=1}\log{\binom{\rho}{m_i}})$ data comparisons on
  $\mathcal{M}$. This number of comparisons is optimal in the worst
  case over multisets of size $n$ formed by $\rho$ runs and $\delta$
  blocks of sizes $g_1, \dots, g_{\delta}$ such that these blocks
  induce a partition $\pi$ of size $\chi$ of the output whose members
  have values $m_1, \dots, m_{\chi}$.
\end{theorem}

We extend these results to take advantage of the global order of the
multiset in a way that can be combined with the notion of runs (local
order).

\subsection{Taking Advantage of Global Order}
\label{sec:global}

Given a multiset $\mathcal{M}$, a \emph{pivot position} is a position $p$ in $\mathcal{M}$ such that all elements in previous position are smaller than or equal to all elements at $p$ or in the following positions. Formally:

\begin{definition}[Pivot positions]
  Given a multiset $\mathcal{M} = (x_1, \dots, x_n)$ of size $n$,
  the pivot positions are the positions $p$ such that $x_a\le x_b$ for
  all $a,b$ such that $a\in[1..p-1]$ and $b\in[p..n]$.
\end{definition}

Existing pivot positions in the input order of $\mathcal{M}$ divide the input into subsequences of consecutive elements such that the range of positions of the elements at each subsequence coincide with the range of positions of the same elements in the sorted sequence of $\mathcal{M}$: the more there are of such positions, the more ``global'' order there is in the input. Detecting such positions takes only a linear number of comparisons.

\begin{lemma}
  Given a multiset $\mathcal{M}$ of size $n$ with $\phi$ pivot
  positions $p_1, \dots, p_{\phi}$, the $\phi$ pivot positions
  can be detected within a linear number of comparisons.
\end{lemma}

\begin{proof}
\begin{TODO}
FACT ~\cite{1973-BOOK-TheArtOfComputerProgrammingVol3-Knuth} does NOT define a step bubble up and/or bubble-down
\end{TODO}
  The \texttt{bubble-up} step of the algorithm
  \texttt{BubbleSort}~\cite{1973-BOOK-TheArtOfComputerProgrammingVol3-Knuth}
  sequentially compares the elements in positions $i-1$ and $i$ of
  $\mathcal{M}$, for $i$ from 2 to $n$. If
  $\mathcal{M}[i-1] > \mathcal{M}[i]$, then the elements interchange
  their values.
  \begin{LONG}
    As consequence of this step the elements with large values tend to
    move to the right.
  \end{LONG}
  In an execution of a \texttt{bubble-up} step in $\mathcal{M}$, the
  elements that do not interchange their values are those elements
  whose values are greater than or equal to all the elements on their
  left. The \texttt{bubble-down} step is similar to the
  \texttt{bubble-up} step, but it scans the sequence from right to
  left\begin{LONG},
    interchanging the elements in positions $i-1$ and $i$ if
    $\mathcal{M}[i-1] > \mathcal{M}[i]$. In an execution of a
    \texttt{bubble-down} step in $\mathcal{M}$, the elements that do
    not interchange their values are those elements whose values are
    smaller than or equal to all the elements on their
    right\end{LONG}. Hence, the positions of the elements that do not
  interchange their values during the executions of both
  \texttt{bubble-up} and \texttt{bubble-down} steps are the pivot
  positions in $\mathcal{M}$. \begin{LONG}\qed\end{LONG}
\end{proof}

When there are $\phi$ such positions, they simply divide the input of
size $n$ into $\phi+1$ sub-instances of sizes $n_0,\ldots,n_\phi$
(such that $\sum^{\phi}_{i=0}n_i = n)$. Each sub-instance $I_i$ for
$i\in[0..\phi]$ then has its own number of runs $r_i$ and alphabet
size $\sigma_i$, on which the synergistic solutions described in this
work can be applied, from mere \textsc{Sorting}
(Section~\ref{sec:synergy-sorting}) to supporting \textsc{MultiSelection} (Section~\ref{sec:synergy-deferr-data}) and  the more sophisticated \textsc{Deferred
Data Structures} (Section~\ref{sec:online-synergy-defer}).

\begin{corollary}
  Let $\mathcal{M}$ be a multiset of size $n$ with $\phi$ pivot
  positions. The $\phi$ pivot positions divide $\mathcal{M}$ into
  $\phi+1$ sub-instances of sizes $n_0,\ldots,n_\phi$ (such that
  $\sum^{\phi}_{i=0} n_i = n$). Each sub-instance $I_i$ of size $n_i$
  is formed by $\rho_i$ runs and $\delta_i$ blocks of sizes
  $g_{i1}, \dots, g_{i\delta_i}$ such that these blocks induce a
  partition $\pi_i$ of the output of size $\chi_i$ whose members have
  values $m_{i1}, \dots, m_{i\chi_i}$ for $i\in[0..\phi]$. There
  exists an algorithm that performs within
  $O(n + \sum^{\phi}_{i=0}\left\{\sum^{\delta_i}_{j=1} \log g_{ij} +
    \sum^{\chi_i}_{j=1}\log{\binom{\rho}{m_{ij}}}\right\})$ data
  comparisons for sorting $\mathcal{M}$. This number of comparisons is
  optimal in the worst case over multisets of size $n$ with $\phi$
  pivot positions which divide the multiset into $\phi+1$ sub-instances of
  sizes $n_0,\ldots,n_\phi$ (such that $\sum^{\phi}_{i=0} n_i = n$)
  and each sub-instance $I_i$ of size $n_i$
  is formed by $\rho_i$ runs and $\delta_i$ blocks of sizes
  $g_{i1}, \dots, g_{i\delta_i}$ such that these blocks induce a
  partition $\pi_i$ of the output of size $\chi_i$ whose members have
  values $m_{i1}, \dots, m_{i\chi_i}$ for $i\in[0..\phi]$.
\end{corollary}

Next, we generalize the algorithm \texttt{Quick Synergy Sort} to an
offline multiselection algorithm that partially sorts a multiset according to the set
of \texttt{select} queries given as input. This serves as a
pedagogical introduction to the online \textsc{Deferred Data
  Structures} for answering \texttt{rank} and \texttt{select} queries
presented in Section~\ref{sec:online-synergy-defer}.

\begin{LONG}
  For simplicity, in Section~\ref{sec:synergy-deferr-data} and
  Section~\ref{sec:online-synergy-defer} we first describe results
  ignoring existing pivot positions, and then present the complete
  result as corollary.
\end{LONG}

\section{MultiSelection Algorithm}\label{sec:synergy-deferr-data}

Given a linearly ordered multiset $\mathcal{M}$ and a sequence of
ranks $r_1, \dots, r_q$, a multiselection algorithm must answer the
queries \texttt{select}($r_1$), $\dots$, \texttt{select}($r_q$) in
$\mathcal{M}$, hence partially sorting $\mathcal{M}$.\begin{INUTILE} We describe
  upper and lower bounds for the multiselection problem based on
  the block representation of the multiset $\mathcal{M}$.

  \subsection{Lower Bound}\label{sec:lower-bound}

  As shown by Dobkin and
  Munro~\cite{1981-JACM-OptimalTimeMinimalSpaceSelectionAlgorithms-DobkinMunro},
  the complexity of the multiselection problem is closely related
  to the complexity of the sorting problem.
  The lower bound for sorting a multiset $\mathcal{M}$, consisting of
  $\rho$ runs and $\delta$ blocks of sizes $g_1, \dots, g_\delta$,
  follows from the results of Demaine et
  al.~\cite{2000-SODA-AdaptiveSetIntersectionsUnionsAndDifferences-DemaineLopezOrtizMunro}.

\begin{corollary}
  Given a multiset $\mathcal{M}$ formed by $\rho$ runs and $\delta$
  blocks of sizes $g_1, \dots, g_\delta$, the entropy \texttt{Sort}(M)
  of the task of sorting $\mathcal{M}$ is
  $n + \delta\log{\rho} + \sum^{\delta}_{i=1} \log g_i$.
\end{corollary}

\begin{TODO}
Carlos, DO CLARIFY: the sentence  introduces definition of pivot blocks and of query blocks - Jeremy
\end{TODO}
We extend the notion of block to the context of partial sorting.
\label{sec:queryBlocks}
\begin{definition}[Query Blocks] 
  Given a multiset $\mathcal{M}$ formed by $\rho$ runs and $\delta$
  blocks, the $q$ \texttt{select} queries correspond to $q$ selected
  elements of $\mathcal{M}$. The ``\emph{query blocks}'' are the
  blocks of $\mathcal{M}$ that contain the selected elements.
\end{definition}

Having determined the $q$ selected elements, the set of elements that
lie between them are known. So, sorting each of these sets yields the
sorted order of the entire multiset.

\begin{lemma}
  Let $\mathcal{M}$ be a multiset formed by $\rho$ runs and $\delta$
  blocks. Let $q$ be the number of \texttt{select} queries. Let
  $G = (g_1, g_2, \dots, g_{\delta})$ be the $\delta$ blocks sorted in
  non-decreasing order. Let $(g'_1, \dots, g'_q)$ be the query blocks
  sorted in non-decreasing order, and let
  $G' = (g'_0 = g_1, g'_1, \dots, g'_q, g'_{q+1} = g_\delta)$ (note
  that $G' \subseteq G$). Let $M_1, \dots, M_{q+1}$ be a partition of
  $G$ such that $M_i$ contains the blocks between $g'_i$ and
  $g'_{i+1}$.  Then, the time complexity of selecting the $q$ queries
  is within
  $\Omega(\texttt{sort}(M) - \sum^{q+1}_{i=1} \texttt{sort}(M_i))$.
\end{lemma}
\end{INUTILE}
We describe a \textsc{MultiSelection} algorithm based on the sorting algorithm \texttt{Quick Synergy Sort} introduced in Section~\ref{sec:ttqu-synergy-sort}. This algorithm is an intermediate result leading to the two \textsc{Deferred Data Structures} described in Section~\ref{sec:online-synergy-defer}.

\begin{INUTILE}
  \subsection{Upper Bound}\label{sec:upper-bound}
\end{INUTILE}

Given a multiset $\mathcal{M}$ and a set of $q$ \texttt{select}
queries, the algorithm \texttt{Quick Synergy MultiSelection}
follows the same first steps as the algorithm \texttt{Quick Synergy
  Sort}. But once it has computed the ranks of all elements in the block that
contains the pivot $\mu$, it determines which \texttt{select} queries
correspond to elements smaller than or equal to $\max_\ell$ and which
ones correspond to elements greater than or equal to $\min_r$ (see
Algorithm~\ref{alg:qss} for the definitions of $\max_\ell$ and
$\min_r$). It then recurses on both sides.
\begin{LONG}
\begin{INUTILE}
    In the process of answering the \texttt{select} queries, the
    algorithm \texttt{Quick Synergy MultiSelection} computes pivots
    that partition the elements of $\mathcal{M}$ into three multisets:
    the elements smaller than or equal to $\max_\ell$, the elements
    greater than or equal to $\min_r$, and the elements between
    $\max_\ell$ and $\min_r$. For this last set the algorithm
    \texttt{Quick Synergy MultiSelection} has already computed the
    ranks of its elements.
  \end{INUTILE}
See
  Algorithm~\ref{alg:qsms} for a formal description of the algorithm  \texttt{Quick
    Synergy MultiSelection}.

  \begin{algorithm} 
    \caption{\texttt{Quick Synergy MultiSelection}} 
    \label{alg:qsms} 
    \begin{algorithmic}[1] 
      \REQUIRE A multiset $\mathcal{M}$ and a set $Q$ of $q$ offline
      \texttt{select} queries \ENSURE The $q$
      selected elements \STATE Compute the $\rho$ runs of respective
      sizes $(r_i)_{i\in[1..\rho]}$ in $\mathcal{M}$ such that
      $\sum^{\rho}_{i=1} r_i = n$; \STATE Compute the
      median $\mu$ of the middles of the $\rho$ runs, note
      $j\in[1..\rho]$ the run containing $\mu$; \STATE Perform
      doubling searches for the value $\mu$ in all runs except the
      $j$-th, starting at both ends of the runs in parallel; \STATE
      Find the maximum $\max_\ell$ (minimum $\min_r$) among the
      elements smaller (resp., greater) than $\mu$ in all runs except
      the $j$-th; \STATE Perform doubling searches for the values
      $\max_\ell$ and $\min_r$ in the $j$-th run, starting at the
      position of $\mu$; \STATE Compute the set of queries $Q_\ell$ that
      go to the left of $\max_\ell$ and the set of queries $Q_r$ that
      go to the right of $\min_r$; \STATE Recurse on the elements
      smaller than or equal to $\max_\ell$ and on the elements greater
      than or equal to $\min_r$ with the set of queries $Q_\ell$ and
      $Q_r$, respectively.
    \end{algorithmic}
  \end{algorithm}
\end{LONG}

We extend the notion of blocks to the context of partial
sorting. The idea is to consider consecutive blocks, which have not been
identified by the \texttt{Quick Synergy
  MultiSelection} algorithm, as a single block. We next introduce the definitions of \emph{pivot blocks} and
\emph{selection blocks}.

\begin{definition}[Pivot Blocks]
  Given a multiset $\mathcal{M}$ formed by $\rho$ runs and $\delta$
  blocks. The ``\emph{pivot blocks}'' are the blocks of $\mathcal{M}$
  that contain the pivots and the elements of value equals to the
  pivots during the steps of the algorithm \texttt{Quick Synergy
    MultiSelection}.
\end{definition}

In each run, between the pivot blocks and the insertion ranks of the
pivots, there are consecutive blocks that the algorithm \texttt{Quick Synergy
  MultiSelection} has not identified as separated blocks, because
no doubling searches occurred inside them.
 
\begin{definition}[Selection Blocks]
  Given the $i$-th run, formed of various blocks, and $q$
  \texttt{select} queries, the algorithm \texttt{Quick Synergy
    MultiSelection} computes $\xi$ pivots in the process of
  answering the $q$ queries. During the doubling searches,
  the algorithm \texttt{Quick Synergy MultiSelection} finds the insertion ranks
  of the $\xi$ pivots inside the $i$-th run. These positions determine
  a partition of size $\xi+1$ of the $i$-th run where each element of
  the partition is formed by consecutive blocks or is empty. We call
  the elements of this partition ``\emph{selection blocks}''. The
  set of all selection blocks include the set of all pivot blocks.
\end{definition}

Using these definitions, we generalize the results proven in Section~\ref{sec:ttqu-synergy-sort} to the more general problem of \textsc{MultiSelection}.

\begin{theorem}\label{theo:qsms}
  Given a multiset $\mathcal{M}$ of size $n$ formed by $\rho$ runs and
  $\delta$ blocks; and $q$ offline \texttt{select} queries over
  $\mathcal{M}$ corresponding to elements of \texttt{ranks} $r_1,
  \dots, r_q$. The algorithm \texttt{Quick Synergy MultiSelection} computes $\xi$ pivots in the process of answering the
  $q$ queries. Let $s_1,\dots, s_{\beta}$ be the sizes of the $\beta$
  selection blocks determined by these $\xi$ pivots in all runs. Let
  $m_1, \dots, m_\lambda$ be the numbers of pivot blocks
  \begin{LONG}
    among this selection blocks
  \end{LONG}
corresponding to the values of the $\lambda$ pivots
with multiplicity greater than 1, respectively.
Let $\rho_0, \dots, \rho_\xi$ be the sequence where $\rho_i$ is the
number of runs that have elements with values between the pivots $i$
and $i+1$ sorted by \texttt{ranks}, for $i\in[1..\xi]$.
  \begin{INUTILE}
    Among these blocks there are $m_1, \dots, m_\lambda$ pivot blocks
    corresponding to the $\chi$ pivots with multiplicity greater than
    1, respectively.
  \end{INUTILE}
  The algorithm \texttt{Quick Synergy MultiSelection} answers
  the $q$ \texttt{select} queries performing within
  $O\left(n + \sum^{\beta}_{i=1}\log{s_i} +
    \beta\log{\rho}-\sum^{\lambda}_{i=1}m_i\log{m_i} -
    \sum^{\xi}_{i=0}\rho_i\log{\rho_i}\right) \subseteq O\left(n\log{n} -
    \sum^{q}_{i=0}\Delta_i\log{\Delta_i}\right)$ data comparisons,
  where $\Delta_i = r_{i+1} - r_i$, $r_0=0$ and $r_{q+1}=n$\begin{INUTILE}
    and
    $O(\beta\log{\rho}) \subseteq O(\min(\xi\rho\log{\frac{n}{\rho}},
    \delta\log{n}))$
  \end{INUTILE}.
\end{theorem}

\begin{proof}
  \begin{VLONG}
    The doubling searches that find the insertion ranks of the pivots
    during the overall execution of the algorithm perform within
    $O(\sum^{\beta}_{i=1}\log{s_i})$ data comparisons
    \begin{LONG}. At each run, a constant factor of the sum of the logarithm of the
      sizes of the selection blocks bounds the number of
      data comparisons performed by these doubling searches
    \end{LONG}
    (see the proof of Lemma~\ref{lem:blocks} analyzing the algorithm
    \texttt{Quick Synergy Sort} for details).
  \end{VLONG}
  
  The pivots computed by the algorithm \texttt{Quick Synergy MultiSelection} for
  answering the queries are a subset of the pivots computed by
  the algorithm \texttt{Quick Synergy Sort} for sorting the whole multiset. Suppose
  that the selection blocks determined by every two consecutive pivots
  form a multiset $\mathcal{M}_j$ such that for every pair of selection
  blocks in $\mathcal{M}_j$ the elements of one are smaller than the
  elements of the other one.\begin{LONG}
    Consider the steps that compute the medians of the middles in the
    algorithm \texttt{Quick Synergy Sort}, the number of
    data comparisons performed by these steps would be within
  \end{LONG}
  \begin{SHORT}
    The algorithm \texttt{Quick Synergy Sort} would perform within
    $O\left(n + \sum^{\beta}_{i=1}\log{s_i} +
      \beta\log{\rho}-\sum^{\lambda}_{i=1}m_i\log{m_i}\right)$ data
    comparisons\end{SHORT}
  \begin{LONG}
    $O\left(n + \sum^{\beta}_{i=1}\log{s_i} +
      \beta\log{\rho}-\sum^{\lambda}_{i=1}m_i\log{m_i}\right)$  
  \end{LONG}
  in this supposed instance (see the proof of\begin{LONG}
    Lemmas~\ref{lem:delta} analyzing the algorithm
    \texttt{Quick Synergy Sort} for details).
  \end{LONG}\begin{SHORT}
    Lemmas~\ref{lem:blocks} and \ref{lem:delta} analyzing the
    algorithm \texttt{Quick Synergy Sort} for details).\end{SHORT}The number of
    comparisons needed to sort the multisets $\mathcal{M}_j$ is within
    $\Theta(\sum^{\xi}_{i=0}\rho_i\log{\rho_i})$.
   The result follows.
  \begin{LONG}\qed\end{LONG}
\end{proof}

The process of detecting the $\phi$ pre-existing pivot positions seen in Section~\ref{sec:global} can be applied as the first step of the multiselection algorithm. The $\phi$ pivot positions divide the input of size $n$ into $\phi+1$ sub-instances of sizes $n_0,\ldots,n_\phi$. For each sub-instance $I_i$ for $i\in[0..\phi]$, the multiselection algorithm determines which \texttt{select} queries correspond to $I_i$ and applies then the steps of the \begin{SHORT}\texttt{Quick Synergy MultiSelection}\end{SHORT}\begin{LONG}Algorithm~\ref{alg:qsms}\end{LONG} inside $I_i$ in order to answer these queries.

\begin{LONG}
  \begin{corollary}
    Let $\mathcal{M}$ be a multiset of size $n$ with $\phi$ pivot
    positions. The $\phi$ pivot positions divide $\mathcal{M}$ into
    $\phi+1$ sub-instances of sizes $n_0,\ldots,n_\phi$ (such that
    $\sum^{\phi}_{i=0} n_i = n$). Let $q$ be the number of offline
    \texttt{select} queries over $\mathcal{M}$, such that $q_i$
    queries correspond to the sub-instance $I_i$, for
    $i\in[0..\phi]$. In each sub-instance $I_i$ of size $n_i$ formed
    by $\rho_i$ runs, the algorithm \texttt{Quick Synergy
      MultiSelection} selects $\xi_i$ pivots when it answers the $q_i$
    queries. These $\xi_i$ pivots determine $\beta_i$ selection blocks
    of sizes $s_{i1}, \dots, s_{i\beta_i}$ inside $I_i$. Let
    $m_{i1}, \dots, m_{i\lambda_i}$ be the numbers of pivot blocks
    \begin{LONG}
      among this selection blocks
    \end{LONG}
    corresponding to the values of the $\lambda_i$ pivots with
    multiplicity greater than 1, respectively. Let
    $\rho_{i0}, \dots, \rho_{i\xi_i}$ be the sequence where $\rho_{ij}$ is the
    number of runs that have elements with values between the pivots
    $ij$ and $i(j+1)$ sorted by \texttt{ranks}, for $j\in[1..\xi_i]$. There
    is an algorithm that answers the $q$ offline \texttt{select}
    queries performing within
    $O(n + \sum^{\phi}_{i=0}\left\{
      \sum^{\beta_i}_{j=1}\log{s_{ij}} +
      \beta_i\log{\rho_i}-\sum^{\lambda_i}_{j=1}m_{ij}\log{m_{ij}} -
      \sum^{\xi_i}_{j=0}\rho_{ij}\log{\rho_{ij}}\right\})$ data comparisons.
  \end{corollary}
\end{LONG}

In the result above, the queries are given all at the same time (i.e.,
offline). In the context where they arrive one at the time (i.e., online), we define
two \textsc{Deferred Data Structures} for answering online
\texttt{rank} and \texttt{select} queries, both inspired by the algorithm
\texttt{Quick Synergy MultiSelection}.

\section{\texttt{Rank} and \texttt{Select} Deferred Data
  Structures}\label{sec:online-synergy-defer}

We describe two \textsc{Deferred Data Structures} that answer a set of \texttt{rank} and \texttt{select} queries arriving one at the time over a multiset $\mathcal{M}$, progressively sorting $\mathcal{M}$.  
Both data structures take advantage of the order (local and global) and structure in the input, and of the structure in the queries.
The first data structure is in the RAM model of computation, at the cost of not taking advantage of the order in which the queries are given. The second data structure is in the comparison model (a more constrained model) but does take advantage of the query order.

\subsection{Taking Advantage of  Order and Structure in the Input, but only of Structure in the Queries}

Given a multiset $\mathcal{M}$ of size $n$, the \textsc{RAM Deferred
  Data Structure} is composed of a bitvector $\mathcal{A}$ of size
$n$, in which we mark the elements in $\mathcal{M}$ that have been
computed as pivots by the algorithm when it answers the online
queries; a dynamic predecessor and successor structure $\mathcal{B}$
over the bitvector $\mathcal{A}$, which allows us to find the two
successive pivots between which the query fits; and for each pivot $p$
found, the data structure stores pointers to the insertion ranks of
$p$ in each run, to the beginning and end of the block $g$ to which
$p$ belongs, and to the position of $p$ inside $g$. The dynamic
predecessor and successor structure $\mathcal{B}$ requires the RAM
model of computation in order to answer \emph{predecessor and
  successor queries} in time within
$o(\log{n})$~\cite{2002-JCSS-OptimalBoundsForThePredecessorProblemAndRelatedProblems-BeameFich}.

\begin{theorem}\label{theo:online-ram}
  Consider a multiset $\mathcal{M}$ of size $n$ formed by $\rho$ runs and
  $\delta$ blocks. The \textsc{RAM Deferred Data Structure} computes
  $\xi$ pivots in the process of answering $q$ online \texttt{rank}
  and \texttt{select} queries over $\mathcal{M}$.
  Let $s_1, \dots, s_{\beta}$ be the sizes of the $\beta$ selection
  blocks determined by these $\xi$ pivots in all runs.
  Let $m_1, \dots, m_\lambda$ be the numbers of pivot blocks
  \begin{LONG}
    among this selection blocks
  \end{LONG}
  corresponding to the values of the $\lambda$ pivots with
  multiplicity greater than 1, respectively.
  Let $\rho_0, \dots, \rho_\xi$ be the sequence where $\rho_i$ is the
  number of runs that have elements with values between the pivots $i$
  and $i+1$ sorted by \texttt{ranks}, for $i\in[1..\xi]$.
  Let $u$ and $g_1, \dots, g_u $ be the number of \texttt{rank}
  queries and the sizes of the identified and searched blocks in the
  process of answering the $u$ \texttt{rank} queries,
  respectively.
  The \textsc{RAM Deferred Data Structure} answers these $q$ online
  \texttt{rank} and \texttt{select} queries in time within
  $O(n + \sum^{\beta}_{i=1}\log{s_i} + \beta\log{\rho} -
  \sum^{\lambda}_{i=1}m_i\log{m_i} -
  \sum^{\xi}_{i=0}\rho_i\log{\rho_i} + \xi\log\log{n} +
  u\log{n}\log\log{n} + \sum^{u}_{i=1}\log{g_i})$.
\end{theorem}

\begin{proof}
  The algorithm answers a new \texttt{select}$(i)$ query by accessing
  in $\mathcal{A}$ the query position $i$. If $\mathcal{A}[i]$ is 1,
  then the element $e$ has been computed as pivot, and hence the
  algorithm answers the query in constant time by following the
  position of $e$ inside the block at which $e$ belongs. If
  $\mathcal{A}[i]$ is 0, then the algorithm finds the nearest pivots
  to its left and right using the predecessor and successor structure,
  $\mathcal{B}$. If the position $i$ is inside a block to which one of
  the two nearest pivots belong, then the algorithm answers the query
  and in turn finishes. If not, it then applies the same steps as the algorithm 
  \texttt{Quick Synergy MultiSelection} in order to answer the
  query; it updates the bitvector $\mathcal{A}$ and the dynamic
  predecessor and successor structure $\mathcal{B}$ whenever a new
  pivot is computed; and for each pivot $p$ computed, the structure
  stores the pointers to the insertion ranks of $p$ in each run, to
  the beginning and end of the block $g$ to which $p$ belongs, and to
  the position of $p$ inside $g$.

  The algorithm answers a new \texttt{rank}$(x)$ query by finding the
  \emph{selection block} $s_j$ in the $j$-th run such that $x$ is
  between the smallest and the greatest value of $s_j$ for all
  $j\in[1..\rho]$. For that, the algorithm performs a sort of parallel
  binary searches for the value $x$ at each run taking advantage of
  the pivots that have been computed by the algorithm. The algorithm
  accesses the position $\frac{n}{2}$ in $\mathcal{A}$. If
  $\mathcal{A}[\frac{n}{2}]$ is 1, then the element $e$ of
  \texttt{rank} $\frac{n}{2}$ has been computed as pivot. Following
  the pointer to the block $g$ to which $e$ belongs, the algorithm
  decides if $x$ is to the right, to the left or inside $g$ by
  performing a constant number of data comparisons. In the last case,
  a binary search for the value $x$ inside $g$ yields the answer of
  the query. If $\mathcal{A}[\frac{n}{2}]$ is 0, then the algorithm
  finds the nearest pivots to the left and right of the position
  $\frac{n}{2}$ using the predecessor and successor structure,
  $\mathcal{B}$. Following the pointers to the blocks that contain
  these pivots the algorithm decides if $x$ is inside one of these
  blocks, to the right of the rightmost block, to the left of the
  leftmost block, or between these two blocks. In the last case, the
  algorithm applies the same steps as the algorithm \texttt{Quick
    Synergy MultiSelection} in order to compute the median $\mu$ of
  the middles and partitions the selection blocks by $\mu$. The
  algorithm then decides to which side $x$ belongs.
  \begin{LONG} These steps identify several
    new pivots, and in consequence several new blocks in the structure.
    \qed
  \end{LONG}
\end{proof}

The \textsc{RAM Deferred Data Structure} includes the pivot positions
(seen in Section~\ref{sec:global}) as a natural extension of the
algorithm. The $\phi$ pivot positions are marked in the bitvector
$\mathcal{A}$. For each pivot position $p$, the structure stores
pointers to the end of the runs detected on the left of $p$; to the
beginning of the runs detected on the right of $p$; and to the
position of $p$ in the multiset.

\begin{LONG}
  \begin{corollary}
    Let $\mathcal{M}$ be a multiset of size $n$ with $\phi$ pivot
    positions. The $\phi$ pivot positions divide $\mathcal{M}$ into
    $\phi+1$ sub-instances of sizes $n_0,\ldots,n_\phi$ (such that
    $\sum^{\phi}_{i=0} n_i = n$). Let $q$ be the number of online
    \texttt{rank} and \texttt{select} queries over $\mathcal{M}$, such
    that $q_i$ queries correspond to the sub-instance $I_i$, for
    $i\in[0..\phi]$.
    In each sub-instance $I_i$ of size $n_i$ formed by $\rho_i$ runs,
    the \textsc{RAM Deferred Data Structure} selects $\xi_i$ pivots in
    the process of answering the $q_i$ online \texttt{rank} and
    \texttt{select} queries over $I_i$.
    Let $s_{i1}, s_{i2}, \dots, s_{i\beta_i}$ be the sizes of the $\beta_i$
    selection blocks determined by the $\xi_i$ pivots in all runs of
    $I_i$.
    Let $m_{1i}, \dots, m_{i\lambda_i}$ be the numbers of pivot blocks
    \begin{LONG}
      among this selection blocks
    \end{LONG}
    corresponding to the values of the $\lambda_i$ pivots with
    multiplicity greater than 1, respectively.
    Let $\rho_{i0}, \dots, \rho_{i\xi_i}$ be the sequence where $\rho_{ij}$ is the
    number of runs that have elements with values between the pivots $ij$
    and $i(j+1)$ sorted by \texttt{ranks}, for $j\in[1..\xi_i]$.
    Let $u_i$ and $g_{i1}, \dots, g_{iu_i}$ be the number of \texttt{rank}
    queries and the sizes of the identified and searched blocks in the
    process of answering the $u_i$ \texttt{rank} queries over $I_i$,
    respectively.
    There exists an algorithm that answers these $q$ online
    \texttt{rank} and \texttt{select} queries in time within
    $O(n + \sum^{\phi}_{i=0}\left\{\beta_i\log{\rho_i} -
      \sum^{\lambda_i}_{j=1}m_{ij}\log{m_{ij}} -
      \sum^{\xi_i}_{j=0}\rho_{ij}\log{\rho_{ij}} + \xi_i\log\log{n_i} +
      u\log{n_i}\log\log{n_i} + \sum^{u_i}_{j=1}\log{g_{ij}}\right\})$.
  \end{corollary}
\end{LONG}

The \textsc{RAM Deferred Data Structure} takes advantage of the
structure in the queries and of the structure and order (local and
global) in the input. Changing the order in the \texttt{rank} and
\texttt{select} queries does not affect the time complexity of the
\textsc{RAM Deferred Data Structure}.\begin{LONG} Once the structure
  identifies the nearest pivots to the left and right of the query
  positions, the steps of the algorithms are the same as in the
  offline case (Section~\ref{sec:synergy-deferr-data}).\end{LONG} We
next describe a deferred data structure taking advantage of the
structure and order in the queries and of the structure and order
(local and global) in the input data.

\subsection{Taking Advantage of the Order and Structure in both the Input and the Queries}

To take advantage of the order in the queries, we introduce a
data structure that finds the nearest pivots to the left and to the
right of a position $p\in[1..n]$, while taking advantage of the
distance between the position of the last computed pivot and $p$. This
distance is measured in the number of computed pivots between the two
positions. For that we use a \emph{finger search
  tree}\begin{SHORT}~\cite{1998-SODA-FingerSearchTreesWithConstantInsertionTime-Brodal}\end{SHORT}\begin{LONG}~\cite{1977-STOC-ANewRepresentationForLinearLists-GuibasMcCreightPlassRoberts}\end{LONG}
which is a search tree maintaining \emph{fingers} (i.e., pointers) to
elements in the search tree. Finger search trees support efficient
updates and searches in the vicinity of the
fingers. Brodal~\cite{1998-SODA-FingerSearchTreesWithConstantInsertionTime-Brodal}
described an implementation of finger search trees that
searches for an element $x$, starting the search at the element given by
the finger $f$ in time within $O(\log{d})$, where $d$ is the distance
between $x$ and $f$ in the set (i.e, the difference between
\texttt{rank}$(x)$ and \texttt{rank}$(f)$ in the set). This operation
returns a finger to $x$ if $x$ is contained in the set, otherwise a
finger to the largest element smaller than $x$ in the set. This
implementation supports the insertion of an element $x$ immediately to
the left or to the right of a finger in \begin{LONG}worst-case\end{LONG} constant
time.

In the description of the \textsc{RAM Deferred Data Structure} from
Theorem~\ref{theo:online-ram}, we substitute the dynamic predecessor
and successor structure $\mathcal{B}$ by a finger search tree
$\mathcal{F}_{\texttt{select}}$, as described by
Brodal~\cite{1998-SODA-FingerSearchTreesWithConstantInsertionTime-Brodal}. Once
a block $g$ is identified, every element in $g$ is a valid pivot for
the rest of the elements in $\mathcal{M}$. In order to capture this
idea, we modify the structure $\mathcal{F}_{\texttt{select}}$ so that
it contains blocks (i.e., a sequence of consecutive values) instead of
singleton pivots. Each element in $\mathcal{F}_{\texttt{select}}$
points in $\mathcal{M}$ to the beginning and the end of the block $g$
that it represents and in each run to the position where the elements
of $g$ partition the run. This modification allows the structure to
answer \texttt{select} queries, taking advantage of the structure and
order in the queries and of the structure and order of the input
data. But in order to answer \texttt{rank} queries taking advantage of
the features in the queries and the input data, the structure needs
another finger search tree $\mathcal{F}_{\texttt{rank}}$. In
$\mathcal{F}_{\texttt{rank}}$ the structure stores for each block $g$
identified, the value of one of the elements in $g$, and pointers in
$\mathcal{M}$ to the beginning and the end of $g$ and in each run to
the position where the elements of $g$ partition the run. We name this
structure \textsc{Full-Synergistic Deferred Data Structure}.
  
   \begin{INUTILE}
     There exists a Deferred Data Structure that answers $q$ online
     \texttt{select} queries over a multiset $\mathcal{M}$ formed by
     $\rho$ runs of sizes $n_1, n_2, \dots, n_\rho$ performing within
     $O(q\sum^{\rho}_{i=1}{\log^2{\frac{n_i}{q}}} +
     q\log{\rho}\sum^{\rho}_{i=1}\log{\frac{n_i}{q}})$ data
     comparisons in the worst case over $q$, $\rho$ and
     $n_1, n_2, \dots, n_\rho$ fix.
   \end{INUTILE}

\begin{theorem}\label{theo:finger}
  Consider a multiset $\mathcal{M}$ of size $n$ formed by $\rho$ runs and
  $\delta$ blocks. The \textsc{Full-Synergistic Deferred Data
    Structure} identifies $\gamma$ blocks in the process of answering
  $q$ online \texttt{rank} and \texttt{select} queries over
  $\mathcal{M}$.
  The $q$ queries correspond to elements of \texttt{ranks} $r_1,
  \dots, r_q$.
  Let $s_1, \dots, s_{\beta}$ be the sizes of the $\beta$
  selection blocks determined by the $\gamma$ blocks in all runs.
  Let $m_1, \dots, m_\lambda$ be the numbers of pivot blocks
  \begin{LONG}
    among this selection blocks
  \end{LONG}
  corresponding to the values of the $\lambda$ pivots with
  multiplicity greater than 1, respectively.
  Let $\rho_0, \dots, \rho_\xi$ be the sequence where $\rho_i$ is the
  number of runs that have elements with values between the pivots $i$
  and $i+1$ sorted by \texttt{ranks}, for $i\in[1..\xi]$.
  Let $d_1, \dots, d_{q-1}$ be the sequence where $d_j$ is the number
  of identified blocks between the block that answers the $j-1$-th
  query and the one that answers the $j$-th query before starting the
  steps to answer the $j$-th query, for $j\in[2..q]$.
  Let $u$ and $g_1, \dots, g_u $ be the number of \texttt{rank} queries and the
  sizes of the identified and searched blocks in the process of
  answering the $u$ \texttt{rank} queries, respectively. The
  \textsc{Full-Synergistic Deferred Data Structure} answers the $q$
  online
  \begin{LONG}
    \texttt{rank} and \texttt{select}
  \end{LONG}
queries performing within
  $O(n + \sum^{\beta}_{i=1}\log{s_i} + \beta\log{\rho} -
  \sum^{\lambda}_{i=1}m_i\log{m_i} -
  \sum^{\xi}_{i=0}\rho_i\log{\rho_i} + \sum^{q-1}_{i=1}\log{d_i} +
  \sum^{u}_{i=1}\log{g_i}) \subseteq O\left(n\log{n} -
    \sum^{q}_{i=0}\Delta_i\log{\Delta_i} + q\log{n}\right)$ data
  comparisons, where $\Delta_i = r_{i+1} - r_i$, $r_0=0$ and $r_{q+1}=n$.
\end{theorem}

\begin{proof}
  The steps for answering a new \texttt{select}$(i)$ query are the
  same as the above description except when the algorithm searches for
  the nearest pivots to the left and right of the query position
  $i$. In this case, the algorithm searches for the position $i$ in
  $\mathcal{F}_{\texttt{select}}$. If $i$ is contained in an element
  of $\mathcal{F}_{\texttt{select}}$, then the block $g$ that contains
  the element in the position $i$ has already been identified. If $i$
  is not contained in an element of $\mathcal{F}_{\texttt{select}}$,
  then the returned finger $f$ points the nearest block $b$ to the
  left of $i$. The block that follows $f$ in
  $\mathcal{F}_{\texttt{select}}$ is the nearest block to the right of
  $i$. Given $f$, the algorithm inserts in
  $\mathcal{F}_{\texttt{select}}$ each block identified in the process
  of answering the query in constant time and stores the respective
  pointers to positions in $\mathcal{M}$.  In
  $\mathcal{F}_{\texttt{rank}}$ the algorithm searches for the value
  of one of the elements in $g$ or $b$. Once the algorithm obtains
  the finger returned by this search, the algorithm inserts in
  $\mathcal{F}_{\texttt{rank}}$ the value of one of the elements of
  each block identified in constant time and stores the
  respective pointers to positions in $\mathcal{M}$.

  The algorithm answers a new \texttt{rank}$(x)$ query by finding the
  \emph{selection block} $s_j$ in the $j$-th run such that $x$ is
  between the smallest and the greatest value of $s_j$ for all
  $j\in[1..\rho]$\begin{LONG}
    , similar to the steps of the \textsc{RAM Deferred Data Structure}
    for answering the query
  \end{LONG}. For that the algorithm searches for the value $x$ in
  $\mathcal{F}_{\texttt{rank}}$. The number of data comparisons
  performed by this searching process is within $O(\log{d})$, where
  $d$ is the number of blocks in $\mathcal{F}_{\texttt{rank}}$ between
  the last inserted or searched block and returned finger $f$. Given
  the finger $f$, there are three possibilities for the \texttt{rank}
  $r$ of $x$: (i) $r$ is between the \texttt{ranks} of the elements at
  the beginning and the end of the block pointed by $f$, (ii) $r$ is
  between the \texttt{ranks} of the elements at the beginning and the
  end of the block pointed by the finger following $f$, or (iii) $r$
  is between the \texttt{ranks} of the elements in the selection
  blocks determined by $f$ and the finger following $f$.  In the cases
  (i) and (ii), a binary search inside the block yields the answer of
  the query. In case (iii), the algorithm applies the same steps as
  the algorithm \texttt{Quick Synergy MultiSelection} in order to compute the
  median $\mu$ of the middles and partitions the selection blocks by
  $\mu$. The algorithm then decides to which side $x$ belongs.
  \begin{LONG}
    These doubling searches identify two new blocks in the
    structure, the block that contains the greatest element smaller
    than or equal to $x$ in $\mathcal{M}$ and the block that contains
    the smallest element greater than $x$ in $\mathcal{M}$.
  \end{LONG}
  \begin{LONG}
    Once compute \texttt{rank}($x$), the algorithm searches for this
    value in $\mathcal{F}_{\texttt{select}}$. It inserts then in
    $\mathcal{F}_{\texttt{select}}$ the block that contains the
    greatest element smaller than or equal to $x$ and
    the block that contains the smallest element greater than $x$. \qed
  \end{LONG}
\end{proof}

The process of detecting the $\phi$ pivot positions seen in
Section~\ref{sec:global} allows the \textsc{Full-Synergistic Deferred
  Data Structure} to insert these pivots in
$\mathcal{F}_{\texttt{select}}$ and $\mathcal{F}_{\texttt{rank}}$. For
each pivot position $p$ in $\mathcal{F}_{\texttt{select}}$ and
$\mathcal{F}_{\texttt{rank}}$, the structure stores pointers to the
end of the runs detected on the left of $p$; to the beginning of the
runs detected on the right of $p$; and to the position of $p$ in the
multiset.\begin{LONG}
  \begin{corollary}
    Let $\mathcal{M}$ be a multiset of size $n$ with $\phi$ pivot
    positions.  The $\phi$ pivot positions divide $\mathcal{M}$ into
    $\phi+1$ sub-instances of size $n_0,\ldots,n_\phi$ (such that
    $\sum^{\phi}_{i=0} n_i = n$).
    Let $q$ be the number of online \texttt{rank} and \texttt{select}
    queries over $\mathcal{M}$, such that $q_i$ queries correspond to
    the sub-instance $I_i$, for $i\in[0..\phi]$.
    In each sub-instance $I_i$ of size $n_i$ formed by $\rho_i$ runs, the
    \textsc{Full-Synergistic Deferred Data Structure} identifies
    $\gamma_i$ blocks in the process of answering $q_i$ online
    \texttt{rank} and \texttt{select} queries over $I_i$.
    Let $s_{i1}, s_{i2}, \dots, s_{i\beta_i}$ be the sizes of the $\beta_i$
    selection blocks determined by the $\gamma_i$ blocks in all runs
    of $I_i$.
    Let $m_{1i}, \dots, m_{i\lambda_i}$ be the numbers of pivot blocks
    \begin{LONG}
      among this selection blocks
    \end{LONG}
    corresponding to the values of the $\lambda_i$ pivots with
    multiplicity greater than 1, respectively.
    Let $\rho_{i0}, \dots, \rho_{i\xi_i}$ be the sequence where $\rho_{ij}$ is the
    number of runs that have elements with values between the pivots $ij$
    and $i(j+1)$ sorted by \texttt{ranks}, for $j\in[1..\xi_i]$.
    Let $d_{i1}, d_{i2}, \dots, d_{iq_{i-1}}$ be the sequence where
    $d_{ij}$ is the number of identified blocks between the block that
    answers the $ij-1$-th query and the one that answers the $ij$-th
    query before starting the steps for answering the $ij$-th query,
    for $j\in[2..q_i]$.
    Let $u_i$ and $g_{i1}, \dots, g_{iu_i}$ be the number of \texttt{rank}
    queries and the sizes of the identified and searched blocks in the
    process of answering the $u_i$ \texttt{rank} queries over $I_i$,
    respectively.
    There exists an algorithm that answers the $q$ online
    \texttt{rank} and \texttt{select} queries performing within
    $O(n + \sum^{\phi}_{i=0} \left\{\sum^{\beta_i}_{j=1}\log{s_{ij}} + \beta_i\log{\rho_i} -
      \sum^{\lambda_i}_{j=1}m_{ij}\log{m_{ij}} -
      \sum^{\xi_i}_{j=0}\rho_{ij}\log{\rho_{ij}} + \sum^{q_i-1}_{j=1}\log{d_{ij}} +
      \sum^{u_i}_{j=1}\log{g_{ij}}\right\})$ data comparisons.
  \end{corollary}
\end{LONG}
\begin{LONG}
  The \textsc{Full-Synergistic Deferred Data Structure} has two
  advantages over the \textsc{RAM Deferred Data Structure}: (i) it is
  in the pointer machine model of computation, which is less powerful
  than the RAM model; and (ii) it takes advantage of the structure and
  order in the queries and of the structure and order (local and
  global) in the input, when the \textsc{RAM Deferred Data Structure}
  does not take advantage of the order in the queries.
\end{LONG}\begin{VLONG}
  Next, we present two compressed data structures, taking advantage of
  the block representation of a multiset $\mathcal{M}$ while
  supporting the operators \texttt{rank} and \texttt{select} on $\mathcal{M}$.
\end{VLONG}\begin{VLONG}
  \section{Compressed Data Structures for \texttt{Rank} and
    \texttt{Select}}
  \label{sec:compressed}

  We describe two compressed representations of a multiset
  $\mathcal{M}$ of size $n$ formed by $\rho$ runs and $\delta$ blocks
  while supporting the operators \texttt{rank} and \texttt{select} on
  it. The first compressed data structure represents $\mathcal{M}$ in
  $\delta\log{\rho} + 3n + o(\delta\log{\rho}+n)$ bits and supports
  each \texttt{rank} query in constant time and each \texttt{select}
  query in time within $O(\log\log\rho)$. The second compressed data
  structure represents $\mathcal{M}$ in
  $\delta\log{\delta} + 2n + O(\delta\log\log{\delta}) + o(n)$ bits
  and supports each \texttt{select} query in constant time and each
  \texttt{rank} query in time within
  $O\left(\frac{\log{\delta}}{\log\log{\delta}}\right)$.

  Given a bitvector $\mathcal{V}$, \texttt{rank}$_1$($\mathcal{V}$,
  $j$) finds the number of occurrences of bit 1 in
  $\mathcal{V}[0..j]$, and \texttt{select}$_1$($\mathcal{V}$, $i$)
  finds the position of the $i$-th occurrence of bit 1 in
  $\mathcal{V}$. Given a sequence $\mathcal{S}$ from an alphabet of
  size $\rho$, \texttt{rank}($\mathcal{S}$, $c$, $j$) finds the number
  of occurrences of character $c$ in $\mathcal{S}[0..j]$;
  \texttt{select}($\mathcal{S}$, $c$, $i$) finds the position of the
  $i$-th occurrence of character $c$ in $\mathcal{S}$; and
  \texttt{access}($\mathcal{S}$, $j$) returns the character at
  position $j$ in $\mathcal{S}$.

  \subsection{Rank-aware Compressed Data Structure for \texttt{Rank}
    and \texttt{Select}}

  The \textsc{Rank-aware Compressed Data
    Structure} \begin{LONG}supports \texttt{rank} in constant time and
    \texttt{select} in time within $O(\log\log{\rho})$ using
    $\delta\log{\rho} + 3n + o(\delta\log{\rho}+n)$
    bits. It \end{LONG}contains three structures $\mathcal{A}$,
  $\mathcal{B}$ and $\mathcal{C}$ representing bitvectors of size $n$
  supporting for
  $\mathcal{V} \in \{\mathcal{A}, \mathcal{B}, \mathcal{C}\}$,
  \texttt{rank}$_1$($\mathcal{V}$, $j$) and
  \texttt{select}$_1$($\mathcal{V}$, $i$) in constant time using
  $n + o(n)$ bits each~\cite{1996-Thesis-CompactPatTrees-Clark}. It
  contains a data structure $\mathcal{S}$ representing a sequence of
  length $\delta$ from an alphabet of size $\rho$ supporting
  \texttt{rank}($\mathcal{S}$, $c$, $j$) in time within
  $O(\log\log{\rho})$, \texttt{access}($\mathcal{S}$, $j$) in time
  within $O(\log\log{\rho})$, and \texttt{select}($\mathcal{S}$, $c$,
  $i$) in constant time, using $\delta\log\rho + o(\delta\log{\rho})$
  bits~\cite{2006-SODA-RankSelectOperationsOnLargeAlphabetsAToolForText-GolynskiMunroRao}. Given
  the blocks $g_1, \dots, g_\delta$ in sorted order, $\mathcal{A}$
  contains the information of the lengths of these blocks in this
  order in a bitvector of length $n$ with a 1 marking the position
  where each block starts. $\mathcal{B}$ contains the information of
  the lengths of the blocks similar to $\mathcal{A}$ but with the
  blocks maintaining the original order, such that all blocks
  belonging to the same run are consecutive. $\mathcal{C}$ contains
  the information of the length of the runs in a bitvector of length
  $n$ with a 1 marking the position where each run starts. The
  structure $\mathcal{S}$ contains the run to which $g$ belongs for
  each block $g$ in sorted order.

\begin{theorem}
  Let $\mathcal{M}$ be a multiset formed by $\rho$ runs and $\delta$
  blocks. The \textsc{Rank-aware Compressed Data Structure} represents
  $\mathcal{M}$ in $\delta\log{\rho} + 3n + o(\delta\log{\rho}+n)$
  bits supporting each \texttt{rank} query in constant time and each
  \texttt{select} query in time within $O(\log\log{\rho})$.
\end{theorem}

\begin{LONG}
  \begin{proof}
    To answer a query \texttt{rank}($\mathcal{M}$, $x$), the following
    operations are executed: \texttt{rank}$_1$($\mathcal{C}$, $i$)
    returns the run $r$ that contains $x$ in constant time, where $i$
    is the position of $x$ in the original order of $\mathcal{M}$;
    \texttt{select}$_1$($\mathcal{C}$, $r$) returns the position $q$
    where $r$ starts in the original order of $\mathcal{M}$ in
    constant time; \texttt{rank}$_1$($\mathcal{B}$, $i$) $-$
    \texttt{rank}$_1$($\mathcal{B}$, $q-1$) returns the position $p$
    inside of $r$ of the block $g$ that contains $x$ in constant time;
    \texttt{select}($\mathcal{S}$, $r$, $p$) returns the position $j$
    of $g$ in sorted order in constant time; and
    \texttt{select}$_1$($\mathcal{A}$, $j$) returns the \texttt{rank}
    of the first element in $g$ in constant time.

    For answering a query \texttt{select}($\mathcal{M}$, $i$), the
    following operations are executed:
    \texttt{rank}$_1$($\mathcal{A}$, $i$) returns the position $j$ of
    the block $g$ in sorted order that contains the selected element
    in constant time; \texttt{access}($\mathcal{S}$, $j$) returns the
    \emph{run} $r$ that contains the selected element in time within
    $O(\log\log\rho)$; \texttt{rank}($\mathcal{S}$, $r$, $j$) returns
    the position $p$ of $g$ inside $r$ in time within
    $O(\log\log\rho)$; and \texttt{select}$_1$($\mathcal{B}$, $p$ $+$
    \texttt{rank}$_1$($\mathcal{B}$,
    \texttt{select}$_1$($\mathcal{C}$, $r$))) returns the position
    where $g$ starts in the original order of $\mathcal{M}$ in
    constant time. \qed
  \end{proof}
\end{LONG}

We describe next a compressed data structure that represents a
multiset, taking advantage of the block representation of it, but
unlike \textsc{Rank-aware Compressed Data Structure}, the structure
supports \texttt{select} in constant time.

\subsection{Select-aware Compressed Data Structure for \texttt{Rank}
  and \texttt{Select}}

The \textsc{Select-aware Compressed Data
  Structure} \begin{LONG}supports \texttt{select} in constant time and
  \texttt{rank} in time within
  $O\left(\frac{\log{\delta}}{\log\log{\delta}}\right)$ using
  $\delta\log{\delta} + 2n + O(\delta\log\log{\delta}) + o(n)$
  bits. It \end{LONG}contains the same two structures $\mathcal{A}$
and $\mathcal{B}$ described above and a structure representing a
permutation $\pi$ of the numbers $[1..\delta]$ supporting the direct
operator $\pi()$ in constant time and the inverse operator
$\pi^{-1}()$ in time within
$O\left(\frac{\log{\delta}}{\log\log{\delta}}\right)$ using
$\delta\log{\delta} + O(\delta\log\log{\delta})$
bits~\cite{2012-TCS-SuccinctRepresentationOfPermutationsAndFunctions-MunroRamanRamanRao}. Given
the blocks $g_1, \dots, g_\delta$ in sorted order, $\pi(i)$ returns
the position $j$ of the block $g_i$ in the original order of
$\mathcal{M}$ and $\pi^{-1}(j) = i$ if the position of the block $g_i$
is $j$ in the original order of $\mathcal{M}$.

\begin{theorem}
  Let $\mathcal{M}$ be a multiset formed by $\rho$ runs and $\delta$
  blocks. The \textsc{Select-aware Compressed Data Structure}
  represents $\mathcal{M}$ in
  $\delta\log{\delta} + 2n + O(\delta\log\log{\delta}) + o(n)$ bits
  supporting each \texttt{select} query in constant time and each
  \texttt{rank} query in time within
  $O\left(\frac{\log{\delta}}{\log\log{\delta}}\right)$.
\end{theorem}

\begin{LONG}
  \begin{proof}
    To answer a query \texttt{select}($\mathcal{M}$, $i$), the
    following operations are executed:
    \texttt{rank}$_1$($\mathcal{A}$, $i$) returns the position $j$ of
    the block $g_j$ in sorted order that contains the selected element
    in constant time; $\pi(j)$ returns the position $p$ of $g_j$ in
    the original order of $\mathcal{M}$ in constant time; and
    \texttt{select}$_1$($\mathcal{B}$, $p$) returns the position where
    $g_j$ starts in $\mathcal{M}$ in constant time.

    To answer a query \texttt{rank}($\mathcal{M}$, $x$), the following
    operations are executed: \texttt{rank}$_1$($\mathcal{B}$, $i$)
    returns the position $j$ of the block $g$ that contains $x$ in
    constant time, where $i$ is the position of $x$ in the original
    order of $\mathcal{M}$; $\pi^{-1}(j)$ returns the position $p$ of
    $g$ in sorted order in time within
    $O\left(\frac{\log{\delta}}{\log\log{\delta}}\right)$; and
    \texttt{select}$_1$($\mathcal{A}$, $p$) returns the \texttt{rank}
    of the first element of $g$ in constant time.\qed
  \end{proof}
\end{LONG}\end{VLONG}This concludes the description of our synergistic results. In the
next section we discuss how these results relate to various past
results and future work.

\section{Discussion}
\label{sec:discussion}

\begin{VLONG}
In the context of deferred data structure, the concept of runs was introduced previously in~\cite{2005-ICALP-TowardsOptimalMultopleSelection-KaligosiMehlhornMunroSanders,2016-JDA-NearOptimalOnlineMultiselectionInInternalAndExternalMemory-BarbayGuptaRaoSorenson}, but for a different purpose than the refined analysis of the complexity presented in this work: we clarify the difference and the research perspectives that it suggests in Section~\ref{sec:previousWork}, and other perspectives for future research in Section~\ref{sec:futureWork}.\begin{LONG} At a higher cognition level, we discuss the importance of categorizing techniques of multivariate analysis of algorithms in Section~\ref{sec:concepts}.\end{LONG}
\end{VLONG}

\begin{VLONG}
\subsection{Comparison with previous work}
\label{sec:previousWork}
\end{VLONG}

Kaligosi et al.'s multiselection
algorithm~\cite{2005-ICALP-TowardsOptimalMultopleSelection-KaligosiMehlhornMunroSanders}
and Barbay et al's deferred data
structure~\cite{2016-JDA-NearOptimalOnlineMultiselectionInInternalAndExternalMemory-BarbayGuptaRaoSorenson}
use the very same concept of \emph{runs} as the one described in this
work.
The difference is that whereas we describe algorithms which
\emph{detect} the existing runs in the input in order to take
advantage of them, the algorithms described by those previous works do
not take into consideration any pre-existing runs in the input
(assuming that there are none) and rather build and maintain such runs
as a strategy to minimize the number of comparisons performed while
partially sorting the multiset.
We leave the combination of both approaches as a topic for future work
which could probably shave a constant factor off the number
of comparisons performed by the \textsc{Sorting} and
\textsc{MultiSelection} algorithms and by the \textsc{Deferred Data Structure}s supporting \texttt{rank} and \texttt{select} on \textsc{Multisets}.
  %
%
\begin{VLONG}
  Johnson and
  Frederickson~\cite{1980-STOC-GeneralizedSelectionAndRanking-FredericksonJohnson}
  described an algorithm answering a single \texttt{select} query in a
  set of sorted arrays of sizes $r_1, r_2, \dots, r_\rho$, in time
  within $O(\sum^{\rho}_{i=1}\log{r_i})$.
  Using their algorithm on pre-existing runs outperforms the Deferred
  Data Structures described in Section~\ref{sec:online-synergy-defer},
  when there is a single query. Yet it is not clear how to generalize
  their algorithm into a deferred data structure in order to support
  more than one query.
  \begin{TODOSRINIVAS}
    GIVE AN INTUITION of WHY here
  \end{TODOSRINIVAS}
  The difference is somehow negligible as the cost of such a query is
  anyway dominated by the cost ($n-1$ comparisons) of partitioning the
  input into runs.
  \begin{LONG}
    We leave the generalization of Johnson and Frederickson's
    algorithm into a deferred data structure which optimally supports
    more than one query as an open problem.
  \end{LONG}\end{VLONG}
  \begin{VLONG}
  We describe additional perspectives for future research in the next section.

  \subsection{Perspectives for future research}
  \label{sec:futureWork}
  \end{VLONG}

\begin{LONG}
  One question to tackle is to see how frequent are ``easy'' instances
  in concrete applications, in terms of input order and structure, and
  in terms of query order and structure; and how much advantage can be
  taken of them.
\end{LONG}


Barbay and Navarro~\cite{2013-TCS-OnCompressingPermutationsAndAdaptiveSorting-BarbayNavarro} described
\begin{LONG}
  how sorting algorithms in the comparison model directly imply
  encodings for permutations, and in particular
\end{LONG}
how sorting algorithms taking advantage of specificities of the input
\begin{SHORT}
  directly
\end{SHORT}
imply compressed encodings of
\begin{LONG}
  such
\end{LONG}
permutations. By using the similarity of the execution tree of the algorithm \texttt{MergeSort} with the
\begin{LONG}
  well known
\end{LONG}
\texttt{Wavelet Tree} data structure, they described a compressed data structure for permutations taking advantage of local order, i.e., using space proportional to $\mathcal{H}(r_1,\ldots,r_\rho)$ and supporting \texttt{direct access} (i.e. $\pi()$) and \texttt{inverse access} (i.e. $\pi^{-1}()$) in worst time within $O(1+\lg\rho)$ and average time within $O(1+\mathcal{H}(r_1,\ldots,r_\rho))$. We leave the definition of a compressed data structure for multisets taking additional advantage of its structure and global order as future work.


Another perspective is to generalize the synergistic results to
related problems in computational geometry: Karp et
  al.~\cite{1988-SIAM-DeferredDataStructuring-KarpMotwaniRaghavan}
  defined the first deferred data structure not only to support
  \texttt{rank} and \texttt{select} queries on multisets, but also to
  support online queries in a deferred way on \textsc{Convex Hull} in
  two dimensions and online \textsc{Maxima} queries on sets of
  multi-dimensional vectors.
One could refine the results from Karp et
al.~\cite{1988-SIAM-DeferredDataStructuring-KarpMotwaniRaghavan}\begin{LONG},
  expressed in function of the number of queries,\end{LONG} to take
into account the blocks between each queries (i.e., the structure in
the queries) as Barbay~et
al.~\cite{2016-JDA-NearOptimalOnlineMultiselectionInInternalAndExternalMemory-BarbayGuptaRaoSorenson}
did for multisets; but also for the relative position of the points
(i.e., the structure in the data) as Afshani~et
al.~\cite{2009-FOCS-InstanceOptimalGeometricAlgorithms-AfshaniBarbayChan}
did for Convex Hulls and Maxima; the order in the points (i.e., the
order in the data), as computing the convex hull in two dimension
takes linear time if the points are sorted; and potentially the order
in the queries.

\begin{LONG}
  \subsection{Importance of the Parameterization of Structure and
    Order}
  \label{sec:concepts}


  The computational complexity of most problems is studied in the
  worst case over instances of fixed size $n$, for $n$ asymptotically
  tending to infinity. This approach was refined for NP-difficult
  problems under the term ``parameterized
  complexity''~\cite{2006-BOOK-ParameterizedComplexityTheory-FlumGrohe},
  for polynomial problems under the term ``Adaptive
  Algorithms''~\cite{1992-ACMCS-ASurveyOfAdaptiveSortingAlgorithms-EstivillCastroWood,1992-ACJ-AnOverviewOfAdaptiveSorting-MoffatPetersson},
  and more simply for data encodings under the term of ``Data
  Compression''~\cite{2013-TCS-OnCompressingPermutationsAndAdaptiveSorting-BarbayNavarro},
  for a wide range of problems and data types.
%
%
  Such a variety of results has motivated various tentative to
  classify them, in the context of NP-hard problems with a theory of
  Fixed Parameter
  Tractability~\cite{2006-BOOK-ParameterizedComplexityTheory-FlumGrohe},
  and in the context of sorting in the comparison model with a theory
  of reduction between
  parameters~\cite{1995-DAM-AFrameworkForAdaptiveSorting-PeterssonMoffat}.
%
%
  We introduced two other perspectives from which to classify
  algorithms and data structures.
  Through the study of the sorting of multisets according to the
  potential ``easiness'' in both the order and the values in the
  multiset, we aimed to introduce a way to classify refined techniques
  of complexity analysis between the ones considering the input order
  and the ones considering the structure in the input; and to show an
  example of the difficulty of combining both into a single hybrid
  algorithmic technique.
  Through the study of the online support of \texttt{rank} and
  \texttt{select} queries on multisets according to the potential
  ``easiness'' in both the order and the values in the queries
  themselves (in addition to the potential easiness in the data being
  queried), we aimed to introduce
%
%
  such categorizations. We predict that such analysis techniques will
  take on more importance in the future, along with the growth of the
  block between practical cases and the worst case over instances of
  fixed sizes. Furthermore, we conjecture that synergistic techniques
  taking advantage of more than one ``easiness'' aspect will be of
  practical importance if the block between theoretical analysis and
  practice is to ever be reduced.
\end{LONG}


\bibliographystyle{splncs}
\bibliography{addedForThePaper,bibliographyDatabaseJeremyBarbay,publicationsJeremyBarbay}

\end{document}